\documentclass[10pt,journal]{IEEEtran}

\usepackage{amsmath, amsfonts, amssymb, amsthm, mathtools}
\usepackage{algorithm, algpseudocode, algorithmicx}

\usepackage{graphicx, pgfplots, tikz, smartdiagram, qcircuit}
\pgfplotsset{compat=1.16}
\usetikzlibrary{calc, positioning, external}
\usesmartdiagramlibrary{additions}

\usepackage{subfig, multirow, xcolor, hyperref, textcomp, url, stfloats, array}
\usepackage{orcidlink}

\usepackage[style=ieee,backend=biber,doi=false,url=false,eprint=false]{biblatex}
\newcommand{\revision}[1]{\textcolor{black}{#1}}

\renewcommand{\epsilon}{\varepsilon}
\newcommand{\bigO}[1]{\mathcal{O}(#1)}

\theoremstyle{plain}
\newtheorem{thm}{Theorem}
\newtheorem{lem}{Lemma}
\newtheorem{definition}{Definition}
\newtheorem{example}{Example}

\newtheorem{prob}{Problem}

\newcommand{\tr}{\operatorname{tr}}
\newcommand{\bra}[1]{\langle #1 |}
\newcommand{\ket}[1]{| #1 \rangle}

\newcommand{\ketbra}[2]{| #1 \rangle \langle #2 |}
\newcommand{\e}{\ensuremath{\mathcal{E}}}
\newcommand{\h}{\ensuremath{\mathcal{H}}}
\renewcommand{\k}{\ensuremath{\mathcal{K}}}
\newcommand{\N}{\ensuremath{\mathcal{N}}}
\renewcommand{\u}{\ensuremath{\mathcal{U}}}

\tikzset{every picture/.append style={remember picture}}
\tikzset{tensor/.style={rectangle, inner sep=2pt,draw,minimum size=0.5cm}}

\newcommand{\tikzinput}[2]{\includegraphics{figures/tikz/#1.pdf}}

\makeatletter
\newcommand{\linebreakand}{%
\end{@IEEEauthorhalign}
\hfill\mbox{}\par
\mbox{}\hfill\hspace{7pt}\begin{@IEEEauthorhalign}
}
\makeatother

\begin{document}

\title{Approximation Methods for Simulation and Equivalence Checking of Noisy Quantum Circuits}

\author{Mingyu Huang\,\orcidlink{0009-0000-3219-1380}, Ji Guan\,\orcidlink{0000-0002-3490-0029}, Wang Fang\,\orcidlink{0000-0001-7628-1185}, and Mingsheng Ying\,\orcidlink{0000-0003-4847-702X}
	\thanks{Mingyu Huang is with Key Laboratory of System Software (Chinese Academy of Sciences) and State Key Laboratory of Computer Science, Institute of Software, Chinese Academy of Sciences, and University of Chinese Academy of Sciences, Beijing 101408, China (e-mail: huangmy@ios.ac.cn). }
	\thanks{Ji Guan is with Key Laboratory of System Software (Chinese Academy of Sciences) and State Key Laboratory of Computer Science, Institute of Software, Chinese Academy of Sciences, Beijing 101408, China (e-mail: guanj@ios.ac.cn). }
	\thanks{Wang Fang is with the School of Informatics, University of Edinburgh, Edinburgh EH8 9UR, United Kingdom (e-mail: wang.fang@ed.ac.uk).}
	\thanks{Mingsheng Ying is with Centre for Quantum Software and Information, University of Technology Sydney, Ultimo NSW 2007, Australia (e-mail: Mingsheng.Ying@uts.edu.au).}
}

\maketitle

\begin{abstract}
	In the current NISQ (Noisy Intermediate-Scale Quantum) era, simulating and verifying noisy quantum circuits is crucial but faces challenges such as quantum state explosion and complex noise representations, constraining simulation and equivalence checking to circuits with a limited number of qubits. This paper introduces an approximation algorithm for simulating and assessing the equivalence of noisy quantum circuits, specifically designed to improve scalability under low-noise conditions. The approach utilizes a novel tensor network diagram combined with singular value decomposition to approximate the tensors of quantum noises. The implementation is based on Google's TensorNetwork Python package for contraction. Experimental results on realistic quantum circuits with realistic hardware noise models indicate that our algorithm can simulate and check the equivalence of QAOA (Quantum Approximate Optimization Algorithm) circuits with around 200 qubits and 20 noise operators, outperforming state-of-the-art approaches in scalability and speed.
\end{abstract}


\section{Introduction}\label{sec:intro}

\IEEEPARstart{S}{ince} achieving quantum supremacy over classical computing \cite{arute2019quantum}, quantum processors with an increasing number of qubits have been manufactured by leading organizations such as Google \cite{ai2024quantum} and IBM \cite{ibm2022quantum, gambetta2023hardware}. This progress marks the transition into the NISQ (Noisy Intermediate-Scale Quantum) era \cite{preskill2018quantum}, which is characterized by devices with an intermediate number of qubits—typically tens to hundreds—that are still susceptible to significant noise. Despite these limitations, NISQ processors have already been used in a variety of applications, demonstrating their potential benefits in fields like optimization, chemistry, material science, and machine learning~\cite{49058, doi:10.1126/science.abb9811, 46227}. At the core of these quantum systems lies the quantum circuit, which represents the essential framework for quantum computation. The design and execution of quantum circuits enable quantum computers to perform calculations that would otherwise be impossible with classical systems. As quantum hardware advances, optimizing the structure, execution, and error resilience of quantum circuits will be crucial to unlocking the full potential of quantum computing. In this context, simulating and performing equivalence checking of quantum circuits have become vital steps for ensuring the correctness and efficiency of quantum computations.

\subsection{Noisy Simulation}
In classical computation, the complexity of testing and verifying digital circuits has necessitated the development of a comprehensive tool chain designed to assist engineers in the rigorous process of ensuring circuit reliability and functionality. Central to this workflow is fault simulation, a critical technique used to model and analyze potential errors in circuit operations. Fault simulation plays a pivotal role in identifying and understanding how various fault models, such as stuck-at faults or bridging faults, can affect the behavior of a circuit. The workflow typically begins with generating a set of specific inputs (test patterns) to detect these faults. These patterns are then applied to a fault simulator, which assesses the circuit's response to each possible fault~\cite{wang2006vlsi}. The results from fault simulation inform the subsequent stages of the workflow, including design for testability (DFT) strategies and automatic test pattern generation (ATPG) as illustrated in Fig.~\ref{fig:workflow}.

\begin{figure}[ht]
	\centering
	\includegraphics[]{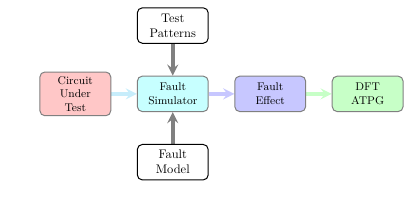}
	\caption{Test Workflow of Classical Circuit}
	\label{fig:workflow}
\end{figure}

In the quantum era, the methodology for testing and verifying circuits is fundamentally redefined, adapting to the unique properties of quantum systems. Quantum testing differs from classical: circuits use unitary gates on qubits (exponential state space), noise models vary by hardware (e.g., decoherence in superconductors), and test patterns are noisy state-preparation circuits. These differences fundamentally distinguish simulating a noisy quantum circuit from its classical counterpart. Therefore, efficient simulation algorithms for noisy quantum circuits need to be developed urgently.

\subsection{Approximate Equivalence Checking}
Equivalence checking plays a pivotal role in electronic design automation (EDA), particularly in verifying the correctness of circuits during the design process. This method is essential in ensuring that a circuit performs its intended function accurately after undergoing various design and optimization stages. By comparing the behavior of a circuit at different levels of abstraction, such as between a high-level functional specification and a lower-level implementation (like Register Transfer Level, or RTL), equivalence checking guarantees that any modifications, enhancements, or optimizations do not inadvertently alter the circuit's fundamental operation. This verification process is indispensable in the development of reliable and efficient electronic systems, as it helps to identify and rectify errors early in the design phase, thereby saving time, reducing costs, and ensuring the final product meets quality and performance standards~\cite{wang2006vlsi}.

Unlike classical circuits, which are characterized by different levels of abstraction, current quantum circuits do not yet have such distinct layers. Nevertheless, the presence of noise in realistic quantum circuits necessitates the evaluation of their performance under specific noise models. Assessing noise impact helps determine if noisy circuits approximate ideal performance. It is particularly valuable to predict the extent to which a real-world quantum circuit, affected by noise, can align with the performance of its ideal, noiseless counterpart. This assessment is key to advancing quantum computing, as it helps design more robust quantum circuits capable of withstanding the practical challenges posed by environmental interference and system imperfections.

\subsection{Contributions of This Paper}
Despite significant advancements in the field of design automation for quantum computing, challenges like the quantum state explosion problem persist. Current simulation and equivalence checking capabilities are limited to real application circuits with tens of qubits. This limitation poses significant issues for Noisy Intermediate-Scale Quantum (NISQ) circuits, which often consist of hundreds of qubits and require efficient simulation and verification tools.

Our key contribution is introducing an approximation algorithm framework based on a tensor network diagram tailored to these tasks. In this framework, noisy quantum circuits are represented as double-size tensor networks, and noise operators are effectively approximated through Singular Value Decomposition (SVD) applied to their tensor representations. Leveraging this decomposition, we develop approximation algorithms for simulation and equivalence checking of noisy quantum circuits and implement it using Google's \texttt{TensorNetwork}~\cite{roberts2019tensornetwork} package.

The efficacy and practicality of our algorithm are demonstrated through extensive experiments on five distinct types of quantum circuits (algorithms). Our results indicate that the proposed method can approximately simulate and verify the equivalence of real application circuits with up to 200 qubits and 20 noise operations. Further evaluation shows that our method effectively handles circuits with realistic NISQ noises, including single-qubit noise and crosstalk. Notably, our approach outperforms state-of-the-art approximation techniques, including the quantum trajectories method~\cite{isakov2021simulations} and Matrix Product Operator (MPO)-based simulation methods, in terms of computational efficiency and scalability.

\subsection{Organization of the Paper:} In Section~\ref{sec:pre} we introduce the basics of quantum computation. In Section~\ref{sec:noisy_simulation}, we give the formal definition of the noisy simulation task and show a tensor network diagram as the basis of our approximation algorithm. The description of our approximate noisy simulation method and the analysis of its precision are presented in Section~\ref{sec:approximation_algorithm}. In Section~\ref{sec:eqc}, we develop an approximate equivalence checking algorithm based on the approximation techniques introduced in Section~\ref{sec:approximation_algorithm}. Finally, we evaluate the runtime and precision of our algorithms in Section~\ref{sec:experiments}.

\subsection{Related Works}

\subsubsection{Simulation of Quantum Circuits}
Simulating quantum circuits is a core task in design automation for quantum computing, particularly for noiseless circuits. The standard approach for simulation involves state vector evolution through matrix multiplication, which is widely implemented in frameworks such as Qiskit and Cirq. However, the exponential growth of state vectors with increasing qubits severely limits the scalability of this approach. To overcome this challenge, tensor network methods leverage the locality (gates operate on a limited number of qubits) and regularity (structured gate patterns) inherent in quantum circuits. These methods enable efficient simulation of large noiseless circuits~\cite{haner20175, huang2020classical, li2019quantum, pednault2017breaking, villalonga2019flexible, PhysRevLett.128.030501}. Decision Diagram (DD)-based methods, inspired by classical Binary Decision Diagrams (BDDs), further optimize memory and performance by compactly representing state amplitudes~\cite{10.1109/DAC18074.2021.9586191, zulehner2018advanced}. However, DD-based methods struggle to handle circuits with arbitrary parameterized gates, such as those found in quantum supremacy experiments~\cite{10.1109/DAC18074.2021.9586191, hillmich2021accurate}.

For noisy quantum circuits, density matrices are used, but scalability is limited by the exponential growth of the Hilbert space with the number of qubits. To mitigate this, approximation techniques such as the quantum trajectories method~\cite{isakov2021simulations, ddsim_noise} have been developed. The quantum trajectories method is a Monte Carlo-based approach for simulating noisy quantum circuits. It works by probabilistically sampling the Kraus operators associated with noise channels and averaging the results obtained from multiple samples. By decomposing noise operators into Kraus matrices, the simulation of each sample reduces to a noiseless quantum circuit simulation task. This allows the method to leverage various noiseless simulation backends. For instance, Google's quantum trajectory-based noisy circuit simulator, qsim~\cite{isakov2021simulations}, employs matrix-vector multiplication as its computational backend. Similarly, a decision diagram (DD)-based backend for quantum trajectories has been implemented in DDSIM~\cite{ddsim_noise}. Tensor network-based methods, such as Matrix Product States (MPS), Matrix Product Operators (MPO), and Matrix Product Density Operators (MPDO), offer another avenue for efficient simulation. These techniques rely on Singular Value Decomposition (SVD) to compress and manipulate quantum states and operators, enabling scalable simulations for certain classes of circuits~\cite{zhou2020limits, woolfe2015matrix, noh2020efficient, cheng2021simulating}.

Despite these advancements, simulating large, complex circuits with significant noise remains challenging. To address these issues, we propose an approximation algorithm based on tensor network diagrams. Our approach achieves greater computational efficiency, especially at low noise levels, compared to the quantum trajectories method and MPO-based techniques.

\subsubsection{Equivalence Checking of Quantum Circuits}
Equivalence checking of combinational quantum circuits has been extensively studied, with methods based on Quantum Multiple-valued Decision Diagrams (QMDDs)~\cite{niemann2014equivalence, burgholzer2020advanced, burgholzer2020improved}. Additionally, Tensor Decision Diagrams (TDD) have been developed to check the equivalence of both combinational and dynamic quantum circuits~\cite{hong2021approximate, hong2022equivalence}, where the latter use measurement outputs to control other parts of the circuit. Most equivalence checking techniques focus on determining the exact equality between two circuits (up to a global phase). In contrast, we focus on approximate equivalence, as proposed in~\cite{hong2021approximate}. Our approach provides a more efficient algorithm for checking the equivalence of noisy quantum circuits, improving performance in practical applications where exact equivalence is not feasible.

\subsubsection{Relation to Our Previous Work}
A preliminary version of our approximation simulation method was presented as a conference paper~\cite{huang2024approximation}. This paper significantly extends that research by applying the method to the equivalence checking task and enhancing existing approximation and equivalence checking techniques. Through additional experiments and analysis, we demonstrate the effectiveness of our approach, providing further evidence of its applicability.

\section{Preliminary}\label{sec:pre}

{\bf Basics of Quantum Computation:}
Quantum data are mathematically modeled as complex unit vectors in a $2^n$-dimensional Hilbert (linear) space $\h$, where $n$ represents the number of involved qubits. A quantum computing task is implemented by a \emph{quantum circuit}, which is mathematically represented by a $2^n\times 2^n$ unitary matrix $U$, i.e., $U^\dagger U=UU^\dagger=I_n$, where $I_n$ is the identity matrix on $\h$. For an input $n$-qubit state $\ket{\psi}$, the output of the circuit is $\ket{\psi'} = U\ket{\psi}.$

Like its classical counterpart, a quantum circuit $U$ consists of a sequence (product) of \emph{quantum logic gates} $U_i$, i.e., $U= U_d\cdots U_1$. Here $d$ is the gate number of the circuit $U$. Each gate $U_i$ only non-trivially operates on one or two qubits.
Fig.~\ref{fig:circuit_example} shows an example of a 2-qubit quantum circuit for the QAOA algorithm.

\begin{figure}
	\begin{align*}
		\begin{array}{c}
			\Qcircuit @C=0.5em @R=0.5em @!R {
			 & \gate{\mathrm{R_Y}\,(\mathrm{\frac{-\pi}{2}})} & \gate{\mathrm{R_Z}\,(\mathrm{\frac{\pi}{2}})} & \ctrl{1}  & \ctrl{1}                      & \gate{\mathrm{R_X}\,(\mathrm{\pi})} & \qw \\
			 & \gate{\mathrm{R_Y}\,(\mathrm{\frac{-\pi}{2}})} & \gate{\mathrm{R_Z}\,(\mathrm{\frac{\pi}{2}})} & \ctrl{-1} & \gate{\mathrm{R_Z}\,(\theta)} & \gate{\mathrm{R_X}\,(\mathrm{\pi})} & \qw
			}
		\end{array}
	\end{align*}
	\caption{A 2-qubit QAOA circuit}
	\label{fig:circuit_example}
\end{figure}
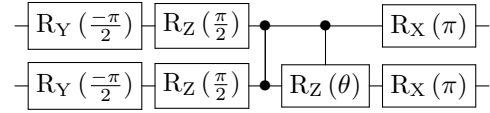

{\bf Quantum Noises:}  In the current NISQ era, quantum noise is unavoidable, so we have to consider the noisy effect in quantum computing.
A \emph{mixed state} is described by a $2^n\times 2^n$ \emph{density matrix} $\rho$ on $\h$:
$\rho = \sum_{k} p_{k}\ket{\psi_{k}}\bra{\psi_{k}}$. In this situation, a computational task starting at a mixed state $\rho$ is finished by a mapping $\e$:
$\rho'=\e(\rho),$ where $\e$ is a quantum channel. A channel $\e$ (also called a \emph{super-operator}) admits a \emph{Kraus matrix form}~\cite{Nielsen2010}: there exists a finite set $\{E_k\}_{k\in\mathcal{K}}$ of matrices on $\h$ such that
\[\setlength\abovedisplayskip{1.5pt}
	\setlength\belowdisplayskip{1.5pt}
	\e(\rho)=\sum_{k\in\k}E_k\rho E_k^\dagger \quad \textrm{ with  } \sum_{k\in\mathcal{K}}E_k^\dagger E_k=I_n,\]
where $\{E_k\}_{k\in\k}$ is called \emph{Kraus matrices} of $\e$. Specifically, quantum gate $U$ can be viewed as a unitary channel where $\e(\rho) = U \rho U^\dagger$. Briefly, $\e$ is represented as $\e=\{E_k\}_{k\in\k}$. An $n$-qubit noise can be described by an $n$-qubit quantum channel with $4^n$ Kraus matrices. A typical single-qubit quantum noise is depolarizing noise, which is represented by the set of operators \( \{\sqrt{1-p}I, \sqrt{\frac{p}{3}}X, \sqrt{\frac{p}{3}}Y, \sqrt{\frac{p}{3}}Z\} \). In this paper, we also adopt a decoherence noise model derived from superconducting circuits~\cite{huang2022fault}.  Mathematically, decoherence manifests as the decay of off-diagonal elements in the density matrix of mixed states, resulting from amplitude damping and phase damping.

\begin{enumerate}
	\item \textbf{Amplitude Damping ($\mathcal{E}_a$):}
	      Amplitude damping describes the gradual loss of energy from excited states. It is characterized by the following Kraus operators:
	      \[
		      E_{a, 1} = \begin{pmatrix} 1 & 0 \\ 0 & e^{-\Delta t/2T_1} \end{pmatrix},
		      E_{a, 2} = \begin{pmatrix} 0 & \sqrt{1 - e^{-\Delta t/T_1}} \\ 0 & 0 \end{pmatrix}
	      \]
	      where $T_1$ is the energy relaxation time and $\Delta t$ is the gate time~\cite{huang2022fault}.

	\item \textbf{Phase Damping ($\mathcal{E}_p$):}
	      Phase damping refers to pure dephasing, where the off-diagonal elements of the density matrix decay without affecting the diagonal elements. It is represented by the following Kraus operators:
	      \[
		      \begin{aligned}
			      E_{p, 0} & = e^{-\Delta t/2T_{\varphi}}I,                             \\
			      E_{p, 1} & = \sqrt{1 - e^{-\Delta t/T_{\varphi}}}|0\rangle\langle 0|, \\
			      E_{p, 2} & = \sqrt{1 - e^{-\Delta t/T_{\varphi}}}|1\rangle\langle 1|,
		      \end{aligned}
	      \]
	      where $T_\varphi$ is defined as $\frac{1}{T_\varphi} = \frac{1}{T_2} - \frac{1}{2T_1}$, and $T_2$ is the total dephasing time~\cite{huang2022fault}.
\end{enumerate}

The overall decoherence noise, denoted by $\mathcal{E}_D$, is the composition of amplitude damping and phase damping:
\[
	\mathcal{E}_D = \mathcal{E}_p \circ \mathcal{E}_a.
\]

\revision{Besides these local noises affecting individual qubits, real quantum devices often exhibit spatially or temporally correlated noises involving multiple qubits. A prevalent example is \emph{crosstalk noise}, which arises from unwanted interactions between neighboring qubits due to imperfect isolation in hardware. Crosstalk can manifest as additional couplings, such as extra ZZ gates, leading to correlated phase errors. For incoherent crosstalk (e.g., when dephased by environmental interactions), it can be represented as a multiple-qubit quantum channel. Although research has modeled multiple-qubit noise models~\cite{breuer2016colloquium} (e.g., crosstalk~\cite{sarovar2020detecting}), detailed crosstalk noise models tailored to realistic devices remain unexplored. Thus, we use two-qubit depolarizing noise as a representative multi-qubit channel for our evaluation.}


Similar to a noiseless quantum circuit \( U \), a noisy quantum circuit \( \mathcal{E} \) also consists of a sequence of noisy gates \( \{\mathcal{E}_i\} \), i.e., \( \mathcal{E} = \mathcal{E}_d \circ \cdots \circ \mathcal{E}_1 \), where each \( \mathcal{E}_i \) is either a noiseless gate or a noise channel.

\section{Noisy Quantum Circuit Simulation}\label{sec:noisy_simulation}

The simulation task for a noiseless quantum circuit involves computing the state vector representation of the circuit's output state given an initial input state. For noisy quantum circuits, the simulation task instead aims to estimate the density matrix of the output state. Specifically, given a noisy quantum circuit \( \mathcal{E}_{\mathcal{N}} \) and an input state \( \rho_0 \), the goal is to estimate \( \mathcal{E}_{\mathcal{N}}(\rho_0) \). The density matrix \( \mathcal{E}_{\mathcal{N}}(\rho_0) \) is a \( 2^n \times 2^n \) complex matrix. However, in practice, it is often unnecessary to store all of its elements explicitly. Instead, many applications require only statistical values obtained from measurements of the state.

In this paper, we focus on simulating the measurement probability of a quantum state, where the measurement is defined by a state \( \ket{v} \), as formalized in Problem~\ref{prob:main}. This setting aligns with current test and verification algorithm frameworks and is widely adopted in hardware experiments. For example, in ATPG of quantum circuits~\cite{chen2024automatic}, quantum gates are discriminated by choosing two states \( \ket{w} \) and \( \ket{w'} \), followed by measurements defined by \( \ketbra{w}{w} \) and \( \ketbra{w'}{w'} \) on the output state separately.

\begin{prob}[Noisy Simulation Task of Quantum Circuits]\label{prob:main}
	Given a noisy quantum circuit with depth \( d \) represented as \( \mathcal{E}_{\mathcal{N}} = \mathcal{E}_d \circ \cdots \circ \mathcal{E}_1 \)~\cite{Nielsen2010}, an input state \( \ket{\psi} \), and a measurement state \( \ket{v} \), the \emph{noisy simulation task} of \( \mathcal{E}_{\mathcal{N}} \) on \( \ket{\psi} \) is to estimate the probability
	\(\bra{v} \mathcal{E}_{\mathcal{N}} (\ketbra{\psi}{\psi}) \ket{v}\) with high accuracy.
\end{prob}

Every entry of \( \mathcal{E}_{\mathcal{N}}(\rho_0) \) can be independently estimated by computing \( \bra{x} \mathcal{E}_{\mathcal{N}}(\rho_0) \ket{y} \), where \( \ket{x}, \ket{y} \) are pure states chosen from the computational basis of \( \mathcal{H} \). Furthermore, solving Problem~\ref{prob:main} allows us to compute these entries efficiently, as they can be expressed as:

\begin{equation}
	\begin{aligned}
		\bra{x} \mathcal{E}_{\mathcal{N}}(\rho_0) \ket{y} = & \frac{1}{4} \Big( (\bra{x} + \bra{y}) \mathcal{E}_{\mathcal{N}}(\rho_0) (\ket{x} + \ket{y}) \\
        & - (\bra{x} - \bra{y}) \mathcal{E}_{\mathcal{N}}(\rho_0) (\ket{x} - \ket{y})                 \\
        & - i (\bra{x} - i\bra{y}) \mathcal{E}_{\mathcal{N}}(\rho_0) (\ket{x} + i\ket{y})             \\
        & + i (\bra{x} + i\bra{y}) \mathcal{E}_{\mathcal{N}}(\rho_0) (\ket{x} - i\ket{y}) \Big).
	\end{aligned}
\end{equation}

We aim to solve Problem~\ref{prob:main} using tensor networks as the data structure for modeling quantum circuits and SVD for approximating the tensor representation of quantum noise. In the following, we introduce a tensor network diagram for the noisy simulation task in Problem~\ref{prob:main}, which is the foundation for our approximation simulation algorithm. Furthermore, this diagram directly leads to a tensor network-based algorithm for exactly computing $\bra{ v }\e_{\N}(\ketbra{\psi}{\psi})\ket{ v }$.

Before introducing our noisy simulation method, we first present a direct computational procedure that demonstrates how tensor networks can be applied to simulate noisy quantum circuits. By leveraging the Kraus representations of quantum noise, the expression $\bra{v}\e_{\N}(\ketbra{\psi}{\psi})\ket{v}$ can be expanded as follows:

\begin{equation}\label{eq:trivial}
\begin{aligned}
  & \left\langle v \left|\e_{\N}\left(\ketbra{\psi}{\psi}\right)\right| v \right\rangle\\
  =\; & \left\langle v \left|\e_d \circ \cdots \circ \mathcal{E}_1\left(\ketbra{\psi}{\psi}\right)\right| v \right\rangle\\
  =\; & \sum_{k_1 \in \mathcal{K}_1} \cdots \sum_{k_d \in \mathcal{K}_d}
        \left\langle v \right|E_{d k_d} \cdots E_{1 k_1}\left| \psi\right\rangle
        \left\langle\psi\right|E_{1 k_1}^{\dagger} \cdots E_{d k_d}^{\dagger}\left| v \right\rangle\\
  =\; & \sum_{k_1, \cdots, k_d }
        \left\langle v \left|E_{d k_d} \cdots E_{1 k_1}\right| \psi\right\rangle
        \left\langle v \left|E_{d k_d} \cdots E_{1 k_1}\right| \psi\right\rangle^*.
\end{aligned}
\end{equation}

This expansion enables a tensor network approach in which one substitutes the Kraus matrices and sums the resulting contractions. However, this method requires $4^{N_1 + 2N_2}$ contractions, where $N_1$ and $N_2$ are the number of single-qubit and two-qubit noises in the simulated noisy circuit, limiting scalability. The exponential cost on $N_1$ and $N_2$ leads to the fact that only noisy quantum circuits with very limited numbers of noise can be handled when simulating NISQ quantum circuits.

	{\bf Tensor Network Diagram for Noisy Simulation:}
Instead of calculating each term separately in Eq.~(\ref{eq:trivial}),
an alternative way to compute $\bra{ v }\e_{\N}(\ketbra{\psi}{\psi})\ket{ v }$ is as follows by using the matrix representation of quantum super-operators~\cite[Chapter 2.2.2]{watrous2018theory}.
\begin{equation*}
	\begin{aligned}
		\bra{ v }\e_{\N}(\ketbra{\psi}{\psi})\ket{ v } & = \tr( \ketbra{v}{v} \e_{\N}(\ketbra{\psi}{\psi}))\\
        & = \bra{\Omega}[\ketbra{ v ^*}{ v ^*}\otimes \e_{\N}(\ketbra{\psi}{\psi})]\ket{\Omega}\\
        & =\bra{\Omega}[\ketbra{ v ^*}{ v ^*}\otimes \e_d\circ\cdots\circ\e_1(\ketbra{\psi}{\psi})]\ket{\Omega}\\
        & = \bra{ v }\otimes\bra{ v ^*}(M_{\e_d}\cdots M_{\e_1}) \ket{\psi}\otimes\ket{\psi^*},
	\end{aligned}
\end{equation*}
where $\ket{\psi^*}$ is the entry-wise conjugate of $\ket{\psi}$, $\ket{\Omega}$ is the (unnormalized) maximal entangled state, i.e., $\ket{\Omega}=\sum_{j}\ket{j}\otimes\ket{j}$ with $\{\ket{j}\}$ being an orthonormal  basis of Hilbert space $\h$, and $M_\e=\sum_{k}E_k\otimes E_k^*$ for $
	\e=\{E_k\}$ is called the \emph{matrix representation} of $\e$. The tensor network notation can visualize this representation. For simplicity, we illustrate this using a one-qubit circuit that consists of only one noise. The trivial calculation method in Eq~(\ref{eq:trivial}) corresponds to the tensor network on the left of Fig.~\ref{fig:alg1}. We can rotate the dashed part of the tensor network and get a tensor network as illustrated on the right side of Fig.~\ref{fig:alg1}.

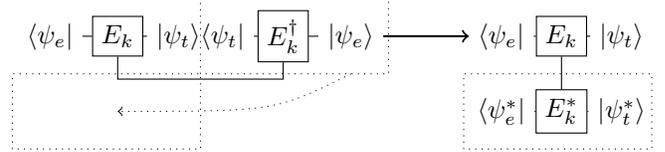
\begin{figure}[!ht]
  \centering
  \begin{tikzpicture}
	\node[draw,rectangle,minimum width=0.5cm, minimum height = 0.5cm] (E) at (0.9,1) {$E_k$};
	\node[draw,rectangle,minimum width=0.5cm, minimum height = 0.5cm] (ED) at (3.1,1) {$E_k^\dagger$};
	\node (D) at (2,1) {$\ketbra{\psi}{\psi}$};
	\node (IL) at (0,1) {$\bra{v}$};
	\node (IR) at (4,1) {$\ket{v}$};
	\draw (IL) -- (E);
	\draw (E) -- (D);
	\draw (D) -- (ED);
	\draw (ED) --(IR);
	\draw[dotted] (2, 0.5) rectangle (4.5, 1.5);
	\draw[dotted] (-0.5, -0.5) rectangle (2, 0.5);
	\draw[dotted, ->] (4, 0.5) .. controls (3, 0) .. (0.9, 0);
	\draw (E.south) -- ++(0, -0.3) -- ++(2.2,0) -- (ED.south);
	
	\node[draw,rectangle,minimum width=0.5cm, minimum height = 0.5cm] (E2) at (6.8,1) {$E_k$};
	\node[draw,rectangle,minimum width=0.5cm, minimum height = 0.5cm] (ED2) at (6.8,0) {$E_k^*$};
	\node (OU) at (7.6,1) {$\ket{\psi}$};
	\node (OD) at (7.6,0) {$\ket{\psi^*}$};
	\node (IU) at (6,1) {$\bra{v}$};
	\node (ID) at (6,0) {$\bra{v^*}$};
	\draw (IU) -- (E2);
	\draw (E2) -- (OU);
	\draw (ID) -- (ED2);
	\draw (ED2) -- (OD);
	\draw (E2.south)  -- (ED2.north);
	\draw[dotted] (5.5, -0.5) rectangle (8, 0.5);
	
	\draw[->,thick] (IR) -- (IU);
	
\end{tikzpicture}
  \caption{Tensor Network Diagram for Noisy Simulation}
  \label{fig:alg1}
\end{figure}

We note that \( M_{\e} = \sum_{k \in \k} E_k \otimes E_k^* \) encodes the effect of quantum noise within the tensor-network representation.
In the special case of a unitary channel \( \u = \{U\} \), this reduces to \( M_{\u} = U \otimes U^* \).
The tensor-network structures corresponding to \( M_{\e} \) and \( M_{\u} \) are illustrated below.
\begin{align*}
	M_{\e} & =\sum_{k\in\k}E_{k}\otimes E_{k}^*=\begin{aligned}
    \Qcircuit @C=1em @R=1em{
     & \gate{E_k}       & \qw  \\
     & \gate{E_k*} \qwx & \qw}\end{aligned} \\
	M_{\u} & =U\otimes U^*=\begin{aligned}
   \Qcircuit @C=1em @R=1em{
    & \gate{U}   & \qw  \\
    & \gate{U^*} & \qw}\end{aligned}
\end{align*}
With this observation, we get a serial connection of the two tensor networks representing $n$-qubit circuits $\e_{\N}$. Subsequently, we can compute $\bra{ v }\otimes\bra{ v ^*}(M_{\e_d}\cdots M_{\e_1}) \ket{\psi}\otimes\ket{\psi^*}$ by contracting a tensor network with double size ($2n$ qubits). An illustration of the tensor network diagram of the two-qubit QAOA circuit mentioned in Sec.~\ref{sec:pre} is shown in Fig.~\ref{fig:double_example}. Based on this idea, an algorithm (Algorithm~\ref{Algorithm:Series}) can be used to compute $\bra{ v }\e_{\N}(\ketbra{\psi}{\psi})\ket{ v }$. When we contract tensor networks, splitting is a standard technique that cuts some edges and converts the task of contracting a large tensor network into contracting a set of smaller tensor networks. Algorithm~\ref{Algorithm:Series} is the same as the direct computational method if we cut all vertical edges of the tensor representation of $M_\e$. However, the double-size tensor network enables tensor network contraction tools to find other more efficient contraction strategies that may be better than  the trivial method in~Eq.(\ref{eq:trivial}).

\begin{figure}
	\centering
	\includegraphics[]{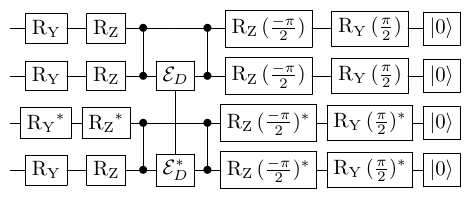}
	\caption{Circuit for noisy simulation of the 2-qubit QAOA circuit in Fig.~\ref{fig:circuit_example} with decoherence noise $\e_{D}$ by appending $\ket{ v } = U\ket{\psi}$ to the end of the circuit, and redundant gates are canceled out.}
	\label{fig:double_example}
\end{figure}

\begin{algorithm}[htbp]
	\caption{NoiseSimulation($\e_{\N},\ket{\psi},\ket{ v }$)}
	\label{Algorithm:Series}
	{\small{
			\begin{algorithmic}[1]
				\Require A noisy quantum circuit $\e_{\N}=\e_d\circ\cdots\circ\e_1$  with $\e_{i}=\{E_{ik}\}_{k\in\k_i}$ for $1\leq i\leq d$, a test state $\ket{\psi}$  and an expected output $\ket{ v }$.
				\Ensure Noisy simulation result $\bra{ v }\e_{\N}(\ketbra{\psi}{\psi})\ket{ v }$
				\State Computing $X=\bra{ v }\otimes\bra{ v ^*}(M_{\e_d}\cdots M_{\e_1}) \ket{\psi}\otimes\ket{\psi^*}$ by calling a tool for tensor network contraction.
				\State\Return X
			\end{algorithmic}
		}}
\end{algorithm}
\section{Approximation Algorithm For Noisy Quantum Circuit Simulation}\label{sec:approximation_algorithm}
We are ready to present our approximation noisy simulation algorithm based on the tensor network diagram with the matrix representation of noises.

	{\bf Approximation Noisy Simulation Algorithm:}
As we can see, the noises make simulating a quantum circuit much harder. To handle the difficulty of simulating a circuit with a large number of noises, we introduce an approximation noisy circuit simulation method based on the matrix representation and SVD, which can balance the accuracy and efficiency of the simulation. The insight of our method starts with an observation that most noises occurring in physical quantum circuits are close to the identity operators (matrices). Take the single-qubit depolarizing noise as an example, which is defined by
\[
	\e(\rho)=(1-p) \rho+\frac{p}{3}(X \rho X+Y \rho Y+Z \rho Z).
\]
Probability $p$ can be regarded as a metric for the effectiveness of the noise channel. When $p$ is small, the noise is almost identical to the identity channel. Therefore, a straightforward method is to approximately represent the depolarizing channel by $(1-p)I$. The current physical implementation ensures that the effectiveness of noise is insignificant.
For general noise $\e$, we define $\|M_\e - I\|$  as the \emph{noise rate} of $\e$, where $\|\cdot\|$ is the 2-norm of matrix. For example, the single-qubit depolarizing noise with parameter $p$ has a noise rate $2p$.

We apply SVD to $M_\e$ for low-rank approximation, exploiting that realistic noises are near-identity. A graphical illustration of this decomposition for single-qubit noise is presented in Fig.~\ref{fig:svd}, where the matrix is reshaped via tensor permutation to facilitate SVD and subsequent truncation. Multi-qubit noises, such as the two-qubit depolarizing channel or crosstalk models introduced in Section~\ref{sec:pre}, can be handled analogously by treating their higher-dimensional matrix representations (e.g., 16×16 for two qubits) as tensors of appropriate rank, applying SVD for decomposition and approximation.

\begin{figure}[!ht]
	\centering
	\subfloat[Tensor Permutation]{
		\includegraphics[width=0.45\textwidth]{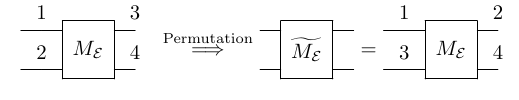}
		\label{subfig:svd1}
	}
	\hfill
	\vspace{-0.5cm} 
	\subfloat[SVD on $\widetilde{M_\epsilon}$]{
		\includegraphics[width=0.45\textwidth]{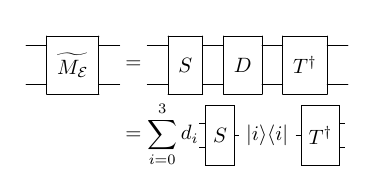}
		\label{subfig:svd2}
	}

	\vspace{-0.5cm} 
	\centering
	\subfloat[Tensor Permutation]{
		\includegraphics[width=0.45\textwidth]{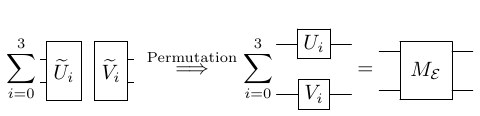}
		\label{subfig:svd3}
	}

	\caption{Decomposition of $M_\epsilon$}
	\label{fig:svd}
\end{figure}

Suppose we treat $M_{\e}$ as a tensor with rank 4 for single-qubit noise, and the edges are indexed as in Fig.~\ref{fig:svd} (a). $M_{\e}$ is the matrix with edges $1, 2$ being the row and $3, 4$ being the column. We introduce the tensor permutation operator where $\tilde{M_{\e}}$ is the matrix with $1,3$ being the row and $2, 4$ being the column. For example, the identity operator $I$ on two qubits  and its tensor permutation $\tilde{I}$ are:
\[
	\begin{array}{c}
		I = \left(\begin{matrix}
				          1 & 0 & 0 & 0 \\
				          0 & 1 & 0 & 0 \\
				          0 & 0 & 1 & 0 \\
				          0 & 0 & 0 & 1
			          \end{matrix}\right), 
	\end{array}
	\begin{array}{c}
		\tilde{I} = \left(\begin{matrix}
				                  1 & 0 & 0 & 1 \\
				                  0 & 0 & 0 & 0 \\
				                  0 & 0 & 0 & 0 \\
				                  1 & 0 & 0 & 1
			                  \end{matrix}\right).
	\end{array}
\]
Performing SVD on $\tilde{M_{\e}}$ we have $\tilde{M_{\e}} = S D T^\dagger$, where $D = \sum\limits_{i=0}^3 d_i \ketbra{i}{i}$ and $d_0$ is the most significant singular value (see Fig.~\ref{fig:svd}(b)). Let $\tilde{U_i} = d_i S\ket{i}$ and $\tilde{V_i} = T\ket{i}$ as in the left of Fig.~\ref{fig:svd} (c), we get $\tilde{M_{\e}} = \sum\limits_{i=0}^3 \tilde{U_i} \tilde{V_i}^\dagger$. Then use tensor permutation again on $\tilde{U_i}$ and $\tilde{V_i}$ (here $\tilde{U_i}$ is treated as a rank 4 tensor with index 3,4 being 0-dimension, and $\tilde{V_i}$ is treated as a rank 4 tensor with index 1,2 being 0-dimension), we get $M_{\e} = \sum\limits_{i = 0, 1, 2, 3} U_i \otimes V_i$, as illustrated in Fig.~\ref{fig:svd} (c).

\begin{example}[Decomposition of Depolarizing Noise]
	Here we explicitly illustrate the calculation steps of our decomposition method on single-qubit depolarizing noise. Specifically, consider depolarizing noise with $p=0.01$, i.e. $\mathcal{E}(\rho)=0.99 \rho+\frac{1}{300}(X \rho X+Y \rho Y+Z \rho Z)$. Our method is as follows:
	\begin{enumerate}
		\item First, from the Kraus matrices of depolarizing noise, we can calculate the matrix representation of the depolarizing noise, which is
		      \begin{align*}
			      M_{\e} & = 0.99I\otimes I + \frac{1}{300}(X\otimes X + Y\otimes Y + Z\otimes Z) \\
			             & = \left(\begin{matrix}
					                       0.9933 & 0      & 0      & 0      \\
					                       0      & 0.9867 & 0.0067 & 0      \\
					                       0      & 0.0067 & 0.9867 & 0      \\
					                       0      & 0      & 0      & 0.9933
				                       \end{matrix}\right).
		      \end{align*}
		\item We start our decomposition by performing tensor permutation first, which derives the matrix
		      \[
			      \widetilde{M_\e} = \left(\begin{matrix}
					      0.9933 & 0      & 0      & 0.9867 \\
					      0      & 0      & 0.0067 & 0      \\
					      0      & 0.0067 & 0      & 0      \\
					      0.9867 & 0      & 0      & 0.9933
				      \end{matrix}\right).
		      \]
            \item By performing SVD on the derived matrix $\widetilde{M_\e}$ as illustrated in Fig.~\ref{fig:svd}(b), we have $\widetilde{M_\e} = S D T^\dagger$, where
            \begin{equation*}
            S = \begin{pmatrix}
            -0.7071 & 0.7071 & 0 & 0 \\
            0 & 0 & -1 & 0 \\
            0 & 0 & 0 & 1 \\
            -0.7071 & -0.7071 & 0 & 0
            \end{pmatrix},
            \end{equation*}
            \begin{equation*}
            D = \begin{pmatrix}
            1.98 & 0 & 0 & 0 \\
            0 & 0.0067 & 0 & 0 \\
            0 & 0 & 0.0067 & 0 \\
            0 & 0 & 0 & 0.0067
            \end{pmatrix},
            \end{equation*}
            \begin{equation*}
            T = \begin{pmatrix}
            -0.7071 & 0 & 0 & -0.7071 \\
            0.7071 & 0 & 0 & -0.7071 \\
            0 & 0 & -1 & 0 \\
            0 & 1 & 0 & 0
            \end{pmatrix}.
            \end{equation*}
		\item Finally, by calculating \( \widetilde{U}_i \) and \( \widetilde{V}_i \), and performing the tensor permutation once more, we complete the decomposition of \( M_{\e} \).
		      \begin{align*}
			      M_{\e} & = 1.98 \left(\begin{matrix} -0.707 & 0 \\ 0 & -0.707 \end{matrix}\right) \otimes \left(\begin{matrix} -0.707 & 0 \\ 0 & -0.707 \end{matrix}\right) \\
			             & + 0.0067 \left(\begin{matrix} 0.707 & 0 \\ 0 & -0.707 \end{matrix}\right) \otimes \left(\begin{matrix} 0.707 & 0 \\ 0 & -0.707 \end{matrix}\right) \\
			             & + 0.0067 \left(\begin{matrix} 0 & -1 \\ 0 & 0 \end{matrix}\right) \otimes \left(\begin{matrix} 0 & 0 \\ -1 & 0 \end{matrix}\right)                 \\
			             & + 0.0067 \left(\begin{matrix} 0 & 0 \\ 1 & 0 \end{matrix}\right) \otimes \left(\begin{matrix} 0 & 1 \\ 0 & 0 \end{matrix}\right).
		      \end{align*}
	\end{enumerate}
\end{example}

Next, we will show that when the noise rate of $\e$ is small, $U_0 \otimes V_0$ is a good approximation of $M_\e$.

\revision{
\begin{lem}
	Suppose $A$ and $B$ are $n \times n$ matrices, and $\tilde{A}$ and $\tilde{B}$ are the tensor permutation of $A$ and $B$ as defined. If $\|A - B\| < \delta$, then we have $\|\tilde{A} - \tilde{B}\| < \sqrt{n} \delta$.
\end{lem}
\begin{proof}
	We use $\|\cdot\|_{\mathrm{F}}$ to denote the Frobenius norm of matrices. For $n \times n$ matrices, it holds that $\|C\| \leq \|C\|_F \leq \sqrt{n} \|C\|$.
	Note that $A$ and $\tilde{A}$ have the same Frobenius norm since the tensor permutation operator only rearranges the positions of elements.
	Therefore we have $\|\tilde{A} - \tilde{B}\| \leq \|\tilde{A} - \tilde{B}\|_F = \|A - B\|_F \leq \sqrt{n} \|A-B\| < \sqrt{n} \delta$.
\end{proof}
\begin{lem}
	Suppose $M_\e$ is $n \times n$ matrix (with $n=4$ for single-qubit noise, $n=16$ for two-qubit noise), and $\|M_\e - I\| < \delta$, then we have $\|M_\e - U_0 \otimes V_0 \| < n \delta$.
\end{lem}
\begin{proof}
	By Lemma 1, we have $\|\tilde{M_\e} - \tilde{I}\| < \sqrt{n} \delta$. Suppose we perform SVD on $\tilde{M_{\e}}$ we have $\tilde{M_{\e}} = S D T^\dagger$. And let $D_0 = d_0 \ketbra{0}{0}$, we have
	\[
		\|\tilde{M_\e} - S D_0 T^\dagger\| = \min_{rank(A) = 1} \| \tilde{M_\e} - A \| \leq \|\tilde{M_\e} - \tilde{I}\| < \sqrt{n} \delta.
	\]
	The first equation comes from the Eckart-Young-Mirsky theorem~\cite{eckart1936approximation}. The last inequality holds because $\tilde{I}$ has rank 1.
	And note that $\tilde{\tilde{M}}_\e = M_\e$ and $\widetilde{S D_0 T^\dagger} = U_0 \otimes V_0$, by Lemma 1, we have
	\[
		\|M_\e -  U_0 \otimes V_0\| < \sqrt{n} \cdot \sqrt{n} \delta = n \delta. \qedhere
	\]
\end{proof}
}
\revision{
We are ready to introduce our approximation algorithm for the noisy quantum circuit simulation task.
Suppose $\e_{\N}=\e_d\circ\cdots\circ\e_1$ is a quantum circuit with $N_1$ single-qubit noises and $N_2$ two-qubit noises, in total $N=N_1+N_2$, where the single-qubit noise channels are indexed by $\{i_1^1,\ldots,i_{N_1}^1\}$ and the two-qubit noise channels by $\{i_1^2,\ldots,i_{N_2}^2\}$.
After applying SVD to the matrix representation of each noise, we obtain for every noise $s \in \{i_1^1,\ldots,i_{N_1}^1\}\cup\{i_1^2,\ldots,i_{N_2}^2\}$ a decomposition
\[
  M_{\e_{s}} \;=\; \sum_{j=0}^{m-1} U_{j}^{\,s} \otimes V_{j}^{\,s},
\]
where $m=4$ if $s$ is single-qubit and $m=16$ if $s$ is two-qubit; moreover $M_{\e_{s}}$ is close to $U_0^{\,s}\otimes V_0^{\,s}$.
The idea of our algorithm is to compute the simulation result by substituting each noise with its dominant term $U_0^{\,s}\otimes V_0^{\,s}$ by default, and then progressively correcting with higher-order terms.
Let $M'_{\e_{s}} := U_0^{\,s}\otimes V_0^{\,s}$ and $\overline{M_{\e_{s}}} := \sum_{j=1}^{m-1} U_{j}^{\,s}\otimes V_{j}^{\,s}$.
Then $M_{\e_{s}} = M'_{\e_{s}} + \overline{M_{\e_{s}}}$ and, by Lemma~2, $\|\overline{M_{\e_{s}}}\| \le c\,\delta$ with $c=4$ for a single-qubit noise and $c=16$ for a two-qubit noise.
Let $T_u$ be the sum of tensor-network terms obtained by substituting all but $u$ noises with $M'_{\e_{s}}$ and the remaining $u$ noises each with one of the residual factors $U_{j}^{\,s}\otimes V_{j}^{\,s}$ ($j=1,\ldots,m-1$), i.e.
\[
  T_u \;=\; \sum_{1 \le p_1 < \cdots < p_u \le N}
        \ \prod_{r=1}^{u} \overline{M_{\e_{i_{p_r}}}}\ \prod_{t \notin \{p_1,\ldots,p_u\}} M'_{\e_{i_t}}.
\]
We call $A(l) := \sum_{u=0}^{l} T_u$ the \emph{$l$-level approximation} of $M_{\e_{\N}}$.
As $l$ increases, the approximation improves, and when $l=N$ we have $A(l)=M_{\e_{\N}}$ exactly.
Our algorithm computes $A(l)$ for a given $l$.
}
\begin{algorithm}[htbp]
\caption{ApproximationNoisySimulation($\e_{\N},\ket{\psi},\ket{ v }, l$)}
\label{Algorithm:Approximate}
{\small{
\begin{algorithmic}[1]
  \Require A noisy quantum circuit $\e_{\N}=\e_d\circ\cdots\circ\e_1$ with $N_1$ single-qubit noises and $N_2$ two-qubit noises, where the single-qubit noise channels are indexed by $\{i_1^1, \ldots, i_{N_1}^1\}$ with Kraus sets $\e_{i_s^1}=\{E_{s\alpha}\}_{\alpha\in\k_s}$ for $1\le s\le N_1$, and the two-qubit noise channels are indexed by $\{i_1^2, \ldots, i_{N_2}^2\}$ with Kraus sets $\e_{i_t^2}=\{E_{t\beta}\}_{\beta\in\k_t}$ for $1\le t\le N_2$; an input state $\ket{\psi}$; a measurement state $\ket{ v }$; and an approximation level $l$.
  \Ensure $G'=\bra{ v }\otimes\bra{ v ^*}A(l) \ket{\psi}\otimes\ket{\psi^*}$, the $l$-level approximation of $\bra{ v }\e_{\N}(\ketbra{\psi}{\psi})\ket{ v }$.
  \State Result $\gets 0$
  \State Generate the double-size tensor network as in Algorithm~\ref{Algorithm:Series}.
  \State For each noise $M_{\e_{s}}$, perform reshuffling and SVD to obtain $U_j^{\,s},V_j^{\,s}$ so that
         $M_{\e_{s}} = \sum_{j=0}^{m_s-1} U_j^{\,s}\otimes V_j^{\,s}$ with $m_s=4$ (single-qubit) or $m_s=16$ (two-qubit).
  \ForAll {$0 \leq u \leq l$} \Comment{approximation level}
    \State $R_u \gets 0$
    \ForAll{$1 \leq p_1 < \cdots < p_u \leq N$}
      \State Substitute $M_{\e_{i_{p_j}}}$ by one residual factor $U_m^{\,i_{p_j}} \otimes V_m^{\,i_{p_j}}$ with $m\in\{1,\ldots,m_{i_{p_j}}-1\}$, and all remaining $M_{\e_{s}}$ by $U_0^{\,s}\otimes V_0^{\,s}$.
      \State The double-size network splits into two $n$-qubit networks; contract both and multiply the scalars to obtain $X$.
      \State $R_u \gets R_u + X$
    \EndFor
    \State Result $\gets$ Result $+ R_u$
  \EndFor
  \State \Return Result
\end{algorithmic}
}}
\end{algorithm}

\revision{
\begin{thm}[Precision and Complexity of Algorithm~\ref{Algorithm:Approximate}]\label{thm: precision}
	Given a noisy quantum circuits $\e_{\N}$ with $d$ gates, $N_1$ single-qubit noises and $N_2$ two-qubit noises with all noise rates being less than $p$, i.e., $\|M_{\e_{i_s}} - I\| < p$ for all noises. For any input and output state $\ket{\psi}$ and $\ket{v}$, let $G = \bra{ v }\e_{\N}(\ketbra{\psi}{\psi})\ket{ v }$, and the approximation result is $G'=\bra{ v }\otimes\bra{ v ^*}A(l) \ket{\psi}\otimes\ket{\psi^*}$, we have $|G - G'| <  \sum_{u = l+1}^{N} \binom{N}{u} (16p)^{u} (1+16p)^{N - u}$, and the number of tensor network contractions is $\bigO{15^lN^l}$.
\end{thm}
\begin{proof} We first observe:
	\begin{equation*}
		\begin{aligned}
			     & \|M_{\e_N} \cdots M_{\e_1} - A(l)\|
			= \|\sum_{u = l+1}^{N} T_u\|\leq \sum_{u = l+1}^{N} \|T_u\|                                             \\
			\leq & \sum_{u = l+1}^{N} \sum_{i_1=0}^{u} \binom{N_1}{i_1} \binom{N_2}{u-i_1} (4p)^{i_1} (16p)^{u-i_1} \\
			     & (1+4p)^{N_1 - i_1} (1+16p)^{N_2 - (u-i_1)}                                                       \\
			< & \sum_{u = l+1}^{N} \sum_{i_1=0}^{u} \binom{N_1}{i_1} \binom{N_2}{u-i_1} (16p)^{u}(1+16p)^{N-u} 	\\
            = & \sum_{u = l+1}^{N} \binom{N}{u} (16p)^{u} (1+16p)^{N - u}     
            \end{aligned}
	\end{equation*}
	Therefore
	\begin{align*}
		     & \|G - G'\|                                                                                               \\
		=    & |\bra{ v }\otimes\bra{ v ^*}(M_{\e_N} \cdots M_{\e_1} - A(l)) \ket{\psi}\otimes\ket{\psi^*}|             \\
		\leq & \|\ket{ v }\otimes\bra{ v ^*}\| \|(M_{\e_N} \cdots M_{\e_1} - A(l)) \ket{\psi}\otimes\ket{\psi^*}\|      \\
		\leq & \|(M_{\e_N} \cdots M_{\e_1} - A(l))\|                                                               &  & \\
		=    & \sum_{u = l+1}^{N} \sum_{i_1=0}^{u} \binom{N_1}{i_1} \binom{N_2}{u-i_1} (4p)^{i_1} (16p)^{u-i_1}         \\
		     & (1+4p)^{N_1 - i_1} (1+16p)^{N_2 - (u-i_1)}
	\end{align*}
	Fixing all but $u$ factors into the dominant form results in a contraction number of $2 \sum_{u = 0}^{l} \sum_{i_1=0}^{u} \binom{N_1}{i_1} \binom{N_2}{u-i_1} 3^{i_1} 15^{u-i_1} \leq 2\sum_{u = 0}^{l} \binom{N}{u} 15^{u}=\bigO{15^lN^l}$.  \qedhere
\end{proof}
}
\revision{
Specifically, when using the one-level approximation where the number of single-qubit noises is $N_1$, two-qubit noises is $N_2$, and the noise rate is $p$, Theorem~\ref{thm: precision} asserts that our method has accuracy $\sum_{u=2}^{N} \sum_{i_1=0}^{u} \binom{N_1}{i_1} \binom{N_2}{u-i_1} (4p)^{i_1} (16p)^{u-i_1} (1+4p)^{N_1 - i_1} (1+16p)^{N_2 - (u-i_1)}$. With the assumption that $p \leq \frac{1}{32N}$, we have:
\[
	\begin{aligned}
		     & \sum_{k=2}^N \sum_{i_1=0}^{k} \binom{N_1}{i_1} \binom{N_2}{k-i_1} (4p)^{i_1} (16p)^{k-i_1} \\
		     & (1+4p)^{N_1 - i_1} (1+16p)^{N_2 - (k-i_1)}                                                 \\
		\leq & \sum\limits_{k=2}^N \tbinom{N}{k}(16p)^k(1+16p)^{N-k}                                      \\
		\leq & (1+\frac{1}{2N})^N \sum\limits_{k=2}^N (16Np)^k                                            \\
		\leq & \sqrt{e} (16Np)^2 \frac{1-(16Np)^{N-1}}{1-16Np}                                            \\
		\leq & 512\sqrt{e}N^2p^2 = \bigO{N^2 p^2}
	\end{aligned}
\]
As a comparison, the quantum trajectories method~\cite{isakov2021simulations, ddsim_noise} needs $r$ samples (tensor network contraction) to obtain a result within accuracy $\bigO{1/\sqrt{r}}$ under a constant probability. For the quantum trajectories method to achieve such accuracy (under a constant probability), $N^2 p^2 = \frac{C}{\sqrt{r}}$. Therefore the sample number $r = \frac{C^2}{N^4 p^4}$, where $C$ is a constant number. Note that our level-1 approximation needs $\bigO{N}$ samples, and therefore our approximation method will need less sampling number than the Monte-Carlo method when $p = \bigO{N^{-\frac{5}{4}}}$.
}
\section{Applications in Equivalence Checking}\label{sec:eqc}
The method developed in Section~\ref{sec:approximation_algorithm} can also be applied to equivalence checking of noisy quantum circuits. Since realistic circuits are affected by noise, instead of determining whether two circuits are exactly equivalent, it is more practical to verify whether they are approximately equivalent to a certain degree. We quantify the degree of equivalence using the Jamiołkowski fidelity between their corresponding super-operators~\cite{gilchrist2005distance}. Let $|\Psi\rangle=\frac{1}{\sqrt{d}}\sum_i |ii\rangle$ be the maximally entangled state. For a super-operator $\mathcal{E}$ on $\mathcal{H}$, define $\rho_{\mathcal{E}}=(\mathcal{I}\otimes\mathcal{E})(|\Psi\rangle\langle\Psi|)$, where $\mathcal{I}$ is the identity super-operator on $\mathcal{H}$.
\begin{definition}[Jamiołkowski fidelity]
  The Jamiołkowski fidelity between super-operators $\mathcal{E}$ and $\mathcal{F}$ is
  \[
    F_J(\mathcal{E}, \mathcal{F}) \;=\; F\!\left(\rho_{\mathcal{E}}, \rho_{\mathcal{F}}\right).
  \]
\end{definition}
As shown in~\cite{hong2021approximate}, when one channel is unitary $\mathcal{U}=\{U\}$, the fidelity admits the closed form
\[
  F_J(\mathcal{U}, \e_{\N})
  \;=\; \frac{1}{2^{2N_q}}\operatorname{Tr}\!\big((U^\dagger\!\otimes\! U^T)\,M_\e\big),
\]
where $N_q$ is the number of qubits of the circuit. This quantity can be computed by contracting a tensor network with twice the number of qubits as the original circuit. Figure~\ref{fig:AEC_example} illustrates the tensor network for the circuit in Fig.~\ref{fig:circuit_example}. Based on this representation, our approximation technique in Section~\ref{sec:approximation_algorithm} yields a practical equivalence-checking algorithm; see Algorithm~\ref{Algorithm:ApproximateEC}. Its precision and complexity are analyzed in Theorem~\ref{thm:AEC}.

\begin{figure}
	\centering
	\includegraphics[]{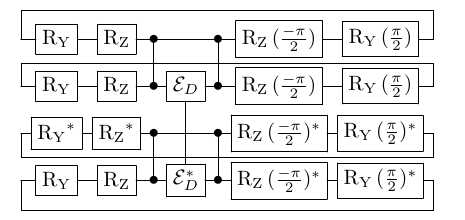}
	\caption{Tensor network for equivalence checking of the 2-qubit QAOA circuit in Fig.~\ref{fig:circuit_example} with decoherence noise $\e_{D}$, and redundant gates are canceled out.}
	\label{fig:AEC_example}
\end{figure}

\begin{algorithm}[htbp]
\caption{ApproximationEquivalenceChecking($\e_{\N}, U, l$)}
\label{Algorithm:ApproximateEC}
{\small{
\begin{algorithmic}[1]
  \Require An ideal circuit $U=U_d\circ\cdots\circ U_1$, and a noisy circuit $\e_{\N}=\e_d\circ\cdots\circ\e_1$ with $N_1$ single-qubit noises and $N_2$ two-qubit noises, where the single-qubit noise channels are indexed by $\{i_1^1,\ldots,i_{N_1}^1\}$ with Kraus sets $\e_{i_s^1}=\{E_{sk}\}_{k\in\k_s}$ for $1\le s\le N_1$, and the two-qubit noise channels by $\{i_1^2,\ldots,i_{N_2}^2\}$ with Kraus sets $\e_{i_t^2}=\{E_{tk}\}_{k\in\k_t}$ for $1\le t\le N_2$.
  \Ensure $F(l)$, the $l$-level approximation of $F_J(\mathcal{U}, \e_{\N})$.
  \State Result $\gets 0$
  \State Generate the tensor network similar to Algorithm~\ref{Algorithm:Approximate}, but connect the inputs and outputs of each doubled wire to take the trace.
  \State For each noise $M_{\e_{i_s}}$, perform reshuffling and SVD to obtain $U_k^s,V_k^s$ so that
         $M_{\e_{i_s^1}} = \sum_{k=0}^{3} U_k^s\otimes V_k^s$ for single-qubit $s\in\{1,\ldots,N_1\}$,
         and $M_{\e_{i_t^2}} = \sum_{k=0}^{15} U_k^t\otimes V_k^t$ for two-qubit $t\in\{1,\ldots,N_2\}$.
  \ForAll {$0 \le u \le l$}
    \State $T_u \gets 0$
    \ForAll{$1 \le p_1 < \cdots < p_u \le N$}
      \State Substitute $M_{\e_{i_{p_j}}}$ with one residual term $U_m^{\,i_{p_j}}\otimes V_m^{\,i_{p_j}}$
             (with $m=1,\ldots,3$ if $i_{p_j}$ is single-qubit, or $m=1,\ldots,15$ if two-qubit),
             and substitute every other $M_{\e_{i_s}}$ with $U_0^{\,s}\otimes V_0^{\,s}$.
      \State The double-sized network splits into two networks; contract both and multiply the results to obtain $X$.
      \State $T_u \gets T_u + X$
    \EndFor
    \State Result $\gets$ Result $+ T_u$
  \EndFor
  \State \Return Result
\end{algorithmic}
}}
\end{algorithm}

\revision{
\begin{thm}[Precision and Complexity of Algorithm~\ref{Algorithm:ApproximateEC}]
	\label{thm:AEC}
Let $\e_{\N}$ be a noisy circuit with $d$ gates, $N_1$ single-qubit noises and $N_2$ two-qubit noises, and suppose all noise rates satisfy $\|M_{\e_{i_s}}-I\|\le p$. Let $N=N_1+N_2$. Then
\[
  \big|F_J(\mathcal{U}, \e_{\N}) - F(l)\big|
  \;\le\; \sum_{u=l+1}^{N} \binom{N}{u} (16p)^{u}\,(1+16p)^{N-u},
\]
and the number of tensor-network contractions is $\bigO{15^lN^l}$.
\end{thm}
\begin{proof}
$$
\begin{aligned}
  & |\operatorname{Tr}((U^{\dagger} \otimes U^{T}) M_{\e})-\operatorname{Tr}((U^{\dagger} \otimes U^{T}) A(l))| \\
= & |\operatorname{Tr}((U^{\dagger} \otimes U^{T})(M_{\e}-A(l)))|.                                               \\
\end{aligned}
$$
By the Cauchy-Schwarz inequality, we have:
$$
\begin{aligned}
     & |\operatorname{Tr}((U^{\dagger} \otimes U^{T})(M_{\e}-A(l)))| \\
=    & |\langle U \otimes U^*, M_{\e}-A(l)\rangle_{HS}|                  \\
\leq & \|U \otimes U^*\|_F \|M_{\e}-A(l)\|_F                        \\
=    & 2^{ N_q}\|M_{\e}-A(l)\|_F
\end{aligned}
$$
Therefore we have
$$
\begin{aligned}
     & |F_J(\mathcal{U}, \mathcal{\e_{\N}}) - F(l)|                                                                                                      \\
\leq & \frac{1}{2^{2N_q}}\cdot 2^{ N} \|M_{\e}-A(l)\|_F                                                                                                      \\
\leq & \frac{1}{2^{N_q}} \|M_{\e}-A(l)\|_F \\
\leq & \|M_{\e}-A(l)\|
\end{aligned}
$$
	Thus, the precision and total contraction number are the same as analyzed in Theorem~\ref{thm: precision}.
\end{proof}
}
\section{Experiments}\label{sec:experiments}
In this section, we demonstrate the utility and effectiveness of our approximation approach. We evaluate our methods by performing noisy simulation and equivalence checking on various quantum circuits with realistic noise models. In particular, we compare our algorithm with state-of-the-art methods for the same tasks, and further numerically analyze our algorithm in terms of efficiency, accuracy, approximation levels, and noise rate. The code of our implementation can be accessed at the GitHub repository. \footnote{\href{https://github.com/Veri-Q/NoiseSim}{https://github.com/Veri-Q/NoiseSim}}.

\subsection{Experiment Configurations:}
\subsubsection{Hardware and Software Tools} All our experiments are carried out on a server with Intel Xeon Platinum 8153 @ $2.00 \mathrm{GHz} \times 256$ Cores, 2048 GB Memory. The machine runs CentOS 7.7.1908. We use the Google TensorNetwork Python package~\cite{roberts2019tensornetwork} for the tensor network computation. The memory out (MO) limit is capped at 2048 GB.

\subsubsection{Benchmark Circuits}
Our benchmark circuits consist of five types of quantum circuits: Bernstein-Vazirani (BV) Algorithm, Quantum Fourier Transform (QFT), Quantum Approximate Optimization Algorithm (QAOA), Hartree-Fock Variational Quantum Eigensolver (VQE), and random quantum circuits exhibiting quantum supremacy from Google (inst). Google has experimentally run the last three types of quantum circuits on their quantum processors. In our experiments, the benchmark circuits for BV and QFT are taken from~\cite{chen2022veriqbench}, and the benchmark circuits for VQE, QAOA and inst are taken from ReCirq~\cite{quantum_ai_team_and_collaborators_2020_4091470}, an open-source library for Cirq and Google's Quantum Computing Service.

\subsubsection{Noise Models}\label{subsubsec:noise_model}
Our method can handle any noise represented by a super-operator. To evaluate its performance, we first consider two single-qubit noise models. The first is depolarizing noise, a fundamental noise model commonly evaluated in most methods. This allows for a fair comparison of performance across different approaches. The other is the decoherence noise explained in Section~\ref{sec:pre}. In our experiments, the default decoherence times are set to $T_1 = 200\,\mu\text{s}$ and $T_2 = 30\,\mu\text{s}$. The decoherence noise is applied after a randomly selected gate in the quantum circuit. While more complex than depolarizing noise, the decoherence noise model approximates the real noise effects in superconducting quantum circuits. 

\revision{To evaluate the performance of our approximation algorithm under NISQ noise models, we compare it with the quantum trajectories method by simulating circuits that incorporate single-qubit depolarizing noise, ZZ-crosstalk noise, and two-qubit depolarizing crosstalk noise, as described in Section~\ref{sec:pre}. The circuits are mapped to the ibm\_brisbane backend, with crosstalk noises applied according to the device topology. In each layer of quantum gates, for any connected qubit pair \(q_i\) and \(q_j\), if a gate is applied to one while the other idles, a ZZ gate with a parameter \(\theta\) drawn uniformly from [-0.1, 0.1] radians is applied to \(q_i\) and \(q_j\). Single-qubit and two-qubit depolarizing noises with \(p=0.0001\) are applied after each respective gate. By evaluating our approximation algorithm on benchmark circuits with these noise models, we obtain a practical assessment of its accuracy and efficiency in simulating the implementation of quantum algorithms on NISQ devices.}

\subsubsection{Parallel Optimization} Our approximation algorithm exhibits inherent parallelism as multiple subtasks of contraction can be executed independently. We implement parallel computing through Python's multiprocessing package to accelerate result computation. For quantum circuits with limited scale (e.g., fewer qubits and gates), this parallelization strategy achieves significant speedup by fully utilizing multi-core CPU resources. However, when handling large-scale quantum circuits, the memory consumption of each subtask approaches the system's physical memory limit, causing diminishing returns from parallelization as computational resources become saturated. While utilizing GPU could enhance performance, the critical memory (VRAM) bottleneck makes CPU-based computation our preferred choice, as it provides substantially larger addressable memory space to support the simulation of larger tensor networks. Finally, we employ a heuristic strategy that adjusts the number of processes based on the number of qubits and gates in the circuit, balancing efficiency with memory limitations.

\subsubsection{Baselines}
To evaluate the performance of our approximation algorithm, we compare it against state-of-the-art methods for simulating and verifying the equivalence of noisy quantum circuits. These baseline methods fall into three categories, each distinguished by its underlying data structure:
\begin{itemize}
	\item \textbf{MM (Matrix Multiplication)-based methods:} MM-based methods represent quantum states, gates, and noise operators as matrices, with simulation performed via matrix multiplications. For this comparison, we adopt Qiskit’s \texttt{evolve} method as the representative MM-based approach.
	\item \textbf{DD (Decision Diagram)-based methods:} DD-based methods utilize tensor-like data structures to encode quantum states, gates, and noise operators. We evaluate two DD-based approaches: (1) the QMDD-based noise-aware simulator DDSIM~\cite{ddsim_noise}, which adapts the stochastic quantum trajectory method from~\cite{ddsim_noise} for noisy circuit simulation, and (2) the Tensor Decision Diagram (TDD)-based method~\cite{hongxinTDD}, modified from~\cite{hong2021approximate}. Additionally, we include the TDD-based equivalence checking tool from~\cite{hong2021approximate}.
	\item \textbf{TN (Tensor Network)-based methods:} Tensor network-based methods are elaborated in Sections~\ref{sec:noisy_simulation},~\ref{sec:approximation_algorithm}, and~\ref{sec:eqc}. Here, 'TN' refers to the implementation of Algorithm~\ref{Algorithm:Series}, while 'Ours' denotes the approximation algorithms introduced in Sections~\ref{sec:approximation_algorithm} and~\ref{sec:eqc}. For comparison, we also implement a basic MPO-based method~\cite{woolfe2015matrix} using the same backend as our proposed methods.
\end{itemize}


\subsection{Experimental Comparison with Baseline Tools:}

\begin{table*}[ht]
	\centering
	\caption{\revision{\textsc{Our Algorithm vs. Accurate Methods (Execution Time)}}}
	\includegraphics[width=0.9\linewidth]{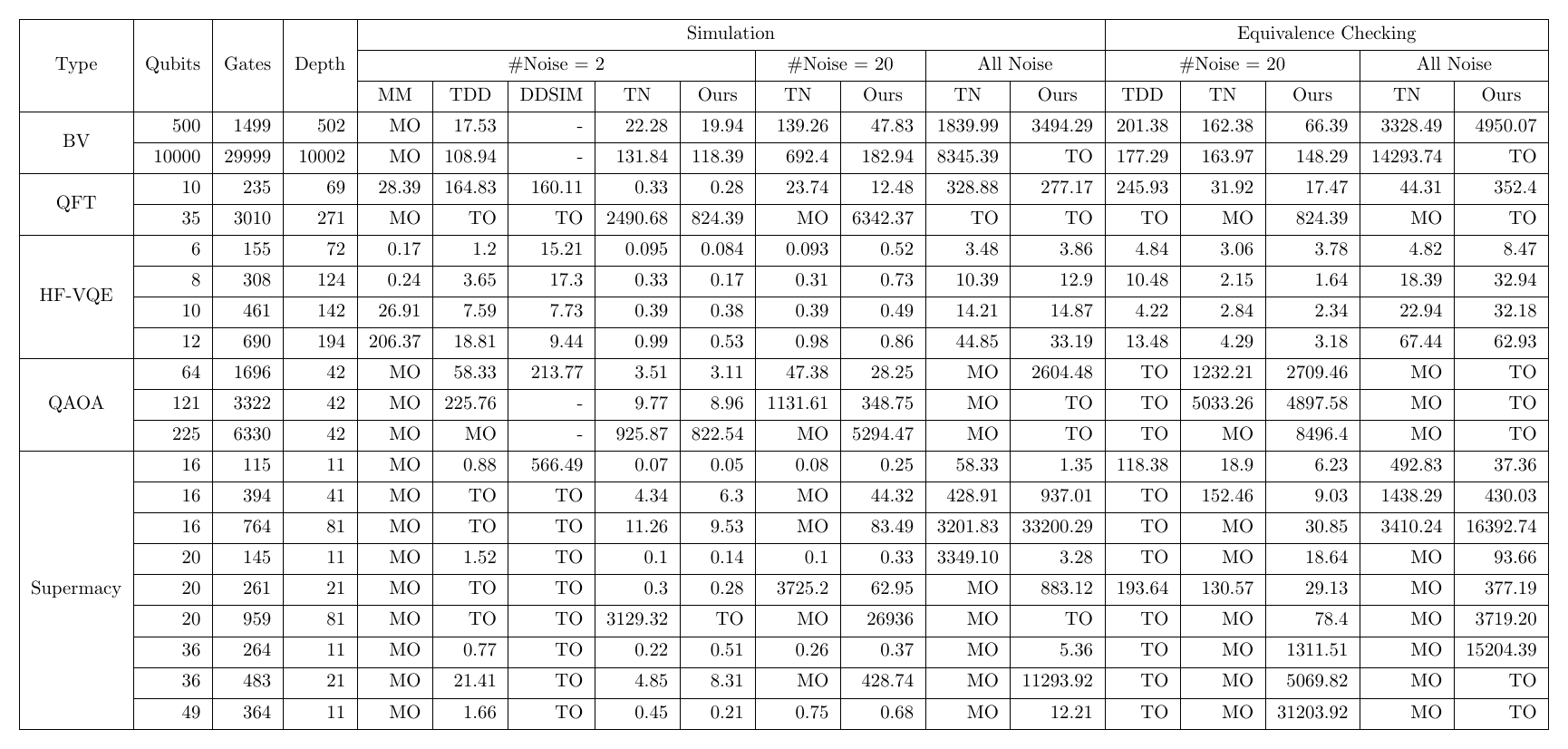}
	\label{table:vqe}
\end{table*}



We compare our simulator with baseline simulators, presenting the results in Tables~\ref{table:vqe} and~\ref{table:mps_compare}. All methods are initially evaluated using a depolarizing noise model with $\#\text{Noise} = 2$, a noise type supported by all simulators, and a timeout threshold of 3,600 seconds. For our method, we employ a level-1 approximation; the DDSIM simulator uses 30,000 samples. The noise-aware DDSIM simulator supports circuits up to 100 qubits; unsupported circuits are denoted by '-'. As shown in Table~\ref{table:vqe}, tensor network (TN)-based methods (TN and ours) outperform matrix multiplication (MM)- and decision diagram (DD)-based methods across most benchmark circuits, including QFT, VQE, QAOA, and Supremacy. Given that our approximation algorithm extends the TN-based approach, their efficiencies are comparable at this low noise level. These findings indicate that TN-based methods are more effective for simulating quantum circuits with practical applications, such as quantum machine learning (VQE and QAOA).

Next, we focus on comparing our method with TN-based simulators under higher noise levels, where TN-based methods struggle. To demonstrate this, we increase the number of noise instances to 20, adopt a decoherence noise model, and extend the timeout threshold to 36,000 seconds to accommodate the increased computational demand. The results, reported in the \#Noise = 20 columns of Table~\ref{table:vqe}, show that while the TN-based method often fails for circuits with larger noise levels, our approximation algorithm maintains strong performance. \revision{We further evaluate the methods on circuits where each gate is followed by a noise operator, representing a realistic scenario on current NISQ devices. In this case, both our method and the TN-based method fail on circuits with a large number of qubits, but our method demonstrates greater capability. We also assess equivalence checking, comparing our method with TDD and TN-based approaches. These findings confirm that our approximation algorithm surpasses other methods in efficiency for circuits with a larger number of noises.}

\revision{We also evaluate the memory consumption of our method compared to the TDD method. Figure~\ref{fig:memory} illustrates the memory usage relative to the percentage of execution time. As shown in the figure, when simulating the 100-qubit QAOA circuit with 20 noise instances, our method consumes less memory during most of the execution period, with memory usage rising toward the end, whereas the TDD method consumes more memory from the beginning and remains high throughout the execution. This demonstrates that the TN backend prioritizes low-cost contractions first and handles the high-cost contractions at the end, making the computation more efficient.}

\begin{figure}[htbp]
\centering
\includegraphics[width=\linewidth]{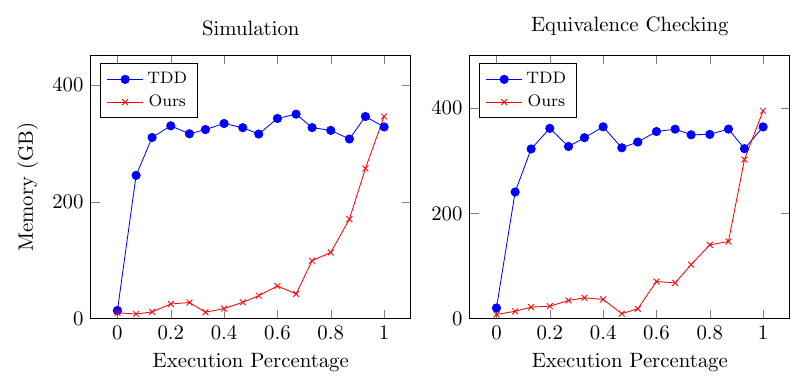}
\caption{\revision{Memory usage of our method and TDD-based method during the simulation equivalence checking}}
\label{fig:memory}	
\end{figure}

Additionally, we compare our approximation method with the MPO-based method, with results detailed in TABLE~\ref{table:mps_compare} for complex quantum circuits. The precision and runtime of the MPO-based method depend on its bond dimension, which we tuned to match the runtime of our algorithm. According to TABLE~\ref{table:mps_compare}, our method achieves higher precision than the MPO-based approach under similar execution times. This indicates that our approximation algorithm offers a superior trade-off between efficiency and precision compared to the MPO-based method.

\subsection{Comparison with the Quantum Trajectories Method:}

Our method is not limited to TN-based simulation. By employing singular value decomposition (SVD) to express noise operators as matrices, we decompose the computational tasks of noisy simulation and equivalence checking into multiple subtasks. These subtasks can be executed using various techniques, allowing our method---akin to the quantum trajectories approach---to offer a versatile approximation scheme compatible with diverse backends. In this section, we compare our approximation scheme with the quantum trajectories method, emphasizing efficiency and precision.

The runtime of the quantum trajectories method depends on the error bound and the probability that the result falls within this bound~\cite{isakov2021simulations, ddsim_noise}. To ensure a fair comparison, we set the success probability of the quantum trajectories method to 99\%. Using Hoeffding's inequality, achieving an accuracy of $\delta$ with a probability exceeding 99\% requires the quantum trajectories method to use $\frac{\ln(100)}{2\delta^2}$ samples, whereas our method, at approximation level 1, needs only $6n + 2$ samples, where $n$ is the noise number. We evaluated the sample requirements for the quantum trajectories method to match the error bound of our algorithm at approximation level 1 across different noise rates $p$. For $p = 0.001$ and $p = 0.0001$, with $n$ ranging from 10 to 40, our method surpasses the quantum trajectories method when $n \leq 26$ at $p = 0.001$, and consistently for all $n$ at $p = 0.0001$ (see Fig.~\ref{fig:monte}). This analysis extends to both simulation and equivalence checking tasks.

\begin{figure}[htbp]
	\begin{minipage}{0.24\textwidth}
		\tikzinput{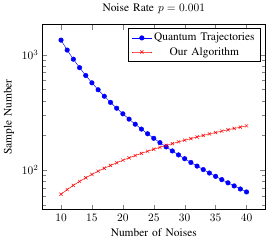}{figures/comparison1.tex}
	\end{minipage}
	\begin{minipage}{0.24\textwidth}
		\tikzinput{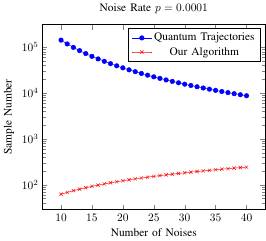}{figures/comparison2.tex}
	\end{minipage}
	\caption{Comparison of sample number required for the same error bound}
	\label{fig:monte}
\end{figure}

Furthermore, we performed numerical experiments to compare the efficiency and precision of our method with the quantum trajectories method, with results presented in Table~\ref{table:traj_num}. The quantum trajectories method was implemented in two variants: an MM-based simulator~\cite{isakov2021simulations} and a TN-based simulator. Using a TN backend for the quantum trajectories method ensures a fair comparison by aligning the underlying framework with our approach. Both experiments employed a depolarizing noise model with a noise number of 20 and a noise rate of $p = 0.001$. We tuned the sample size of the quantum trajectories method to achieve the same precision as our level-1 approximation. As shown in Table~\ref{table:traj_num}, our method outperforms the quantum trajectories method (denoted "Traj") in efficiency at equivalent precision levels. Notably, our algorithm surpasses the TN-based quantum trajectories method, despite sharing the same backend, highlighting the superior efficiency of our underlying scheme.


\begin{table*}[ht]
	\centering
	\caption{\textsc{Our Algorithm v.s. MPO-based Methods}}
	\includegraphics[width=\textwidth]{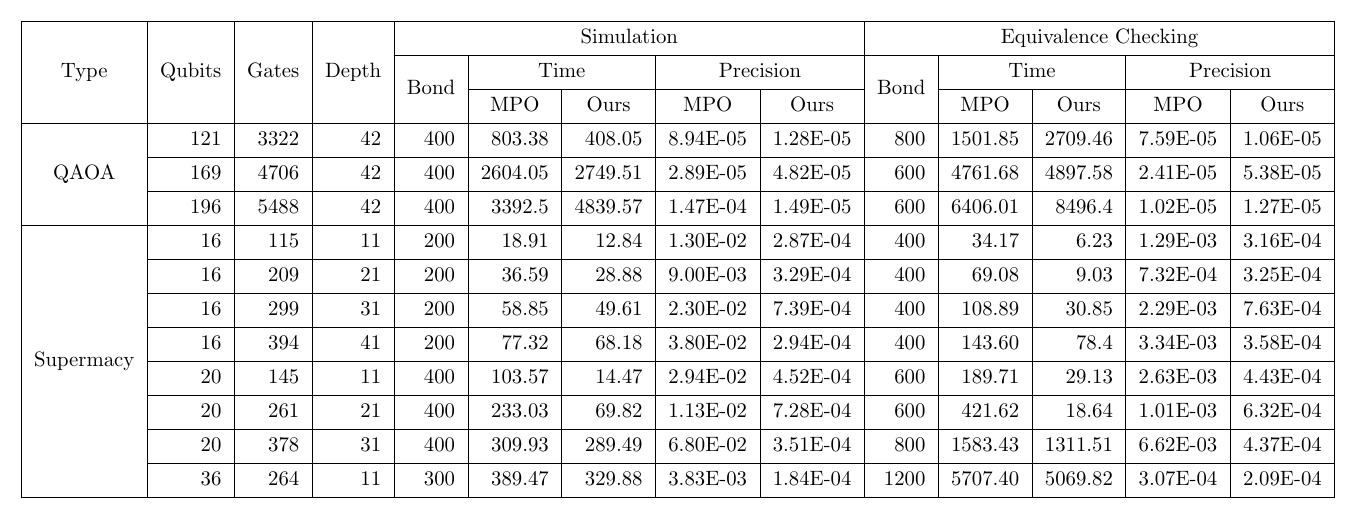}
	\label{table:mps_compare}
\end{table*}
\begin{table}
\centering
\caption{\textsc{Our Algorithm v.s. Approximate Methods}}
\label{table:traj_num}
\resizebox{\columnwidth}{!}{
\begin{tabular}{|c|r|r|r|r|r|r|}
\hline
\multirow{2}{*}{Circuit} & \multicolumn{3}{c|}{Precision}                                                              & \multicolumn{3}{c|}{Runtime}                                                                \\
\cline{2-7}
                         & \multicolumn{1}{c|}{{Ours}} & \multicolumn{1}{c|}{Traj (MM)} & \multicolumn{1}{c|}{Traj (TN)} & \multicolumn{1}{c|}{{Ours}} & \multicolumn{1}{c|}{Traj (MM)} & \multicolumn{1}{c|}{Traj (TN)} \\
\hline
QAOA\_6                  & {1.55E-04}                  & 1.43E-04                       & 1.80E-04                       & {4.56}                      & 1.49                           & 10.84                          \\
\hline
QAOA\_10                 & {5.96E-05}                  & 1.98E-04                       & 2.00E-04                       & {7.57}                      & 10.94                          & 19.68                          \\
\hline
QAOA\_15                 & {9.03E-05}                  & 2.55E-04                       & 2.66E-04                       & {8.41}                      & 658.38                         & 25.48                          \\
\hline
QAOA\_64                 & {5.82E-07}                  & MO                             & 3.48E-07                       & {191.47}                    & MO                             & 598.62                         \\
\hline
QAOA\_100                & {1.03E-07}                  & MO                             & 1.30E-07                       & {361.79}                    & MO                             & 1264.67                        \\
\hline
\end{tabular}%
} 
\end{table}

\begin{table}
\centering
\caption{\textsc{Accuracy for Different Approximation Levels}}
\label{table:merged_approximate_level}
\begin{tabular}{|c|r|r|r|r|} 
\hline
\multirow{2}{*}{Level} & \multicolumn{2}{c|}{Simulation}                            & \multicolumn{2}{c|}{Equivalence Checking}                   \\ 
\cline{2-5}
                       & \multicolumn{1}{c|}{Time (s)} & \multicolumn{1}{c|}{Error} & \multicolumn{1}{c|}{Time (s)} & \multicolumn{1}{c|}{Error}  \\ 
\hline
0                      & 0.34                          & 4.59E-03                   & 1.74                          & 1.33E-02                    \\ 
\hline
1                      & 11.18                         & 3.02E-05                   & 39.48                         & 8.86E-05                    \\ 
\hline
2                      & 109.95                        & 1.23E-06                   & 183.96                        & 3.60E-06                    \\ 
\hline
3                      & 1971.37                       & 1.13E-06                   & 3029.5                        & 3.28E-06                    \\
\hline
\end{tabular}
\end{table}

\subsection{Results on Noise}

\subsubsection{Results on Noise Number} Furthermore, to highlight the high efficiency of our approximation algorithm with respect to the number of noise instances, we simulate and check the equivalence of the 100-qubit QAOA circuit with 0 to 80 noise instances, as shown in Fig.~\ref{fig:runtime}. Our approximation algorithm handles all cases, while the TN-based method runs out of memory once the number of noise instances reaches 30. The primary reason for this memory failure is that increasing the number of noise instances can increase the maximum rank of the nodes in the tensor network contraction, leading to higher time and memory consumption during contraction. In comparison, when using the level-1 approximation of Algorithm~\ref{Algorithm:Approximate}, the runtime of our approximation algorithm increases almost linearly with the number of noise instances, as shown in Fig.~\ref{fig:runtime}, which is consistent with Theorem.~\ref{thm: precision}.

\begin{figure}[!htbp]
	\centering
	\includegraphics[width=0.5\textwidth]{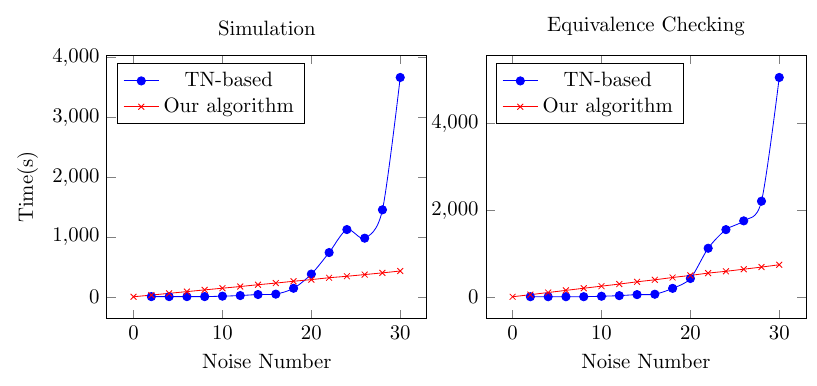}
	\caption{The number of noise and runtime}
	\label{fig:runtime}
\end{figure}

\subsubsection{Results on Noise Rate} The accuracy of our simulation and equivalence checking algorithm depends on the noise rate of the noises, i.e., $p = \|M_{\e_D} - I\|$. When $p$ is small, the algorithm has high accuracy. We evaluated our algorithms under both the realistic model and the depolarizing noise model. In both cases, the results as shown in Fig.~\ref{fig:noise_rate_sim} and~\ref{fig:noise_rate_eqc} demonstrate that as the noise rate increases, the approximation error also rises. It confirms that our approximation algorithm has better accuracy for lower noise rates in realistic noise models, which indicates a promising outlook for achieving higher precision on advanced hardware implementations in the future.

\subsubsection{Results on Approximation Levels} Our algorithms offers varying levels of approximation, presenting a trade-off between computational efficiency and accuracy. As a demonstration, we simulate and check equivalence of the 64-qubit QAOA circuit with 10 noises. In simulation task the input state $\ket{\psi} = \ket{0}\otimes\cdots\otimes\ket{0}$ and $\ket{v}=U\ket{0}\otimes\cdots\otimes\ket{0}$, where $U$ represents the unitary operator of the ideal circuit. Table.~\ref{table:merged_approximate_level} shows the results for the same circuit with approximation levels from 0 to 3. From the result, we can see that the accuracy is getting higher as expected, but the cost for an extra level of approximation is also significant. For most scenarios, the level-1 approximation is recommended since it can achieve a good balance between runtime and accuracy.



\begin{table*}[!htb]
    \centering
    \caption{\textsc{NISQ Noise Simulation and Equivalence Checking}}
    \label{table:crosstalk}
    \begin{tabular}{|c|r|r|r|r|r|r|l|r|r|r|}
        \hline
        \multirow{3}{*}{Circuit} & \multicolumn{5}{c|}{Simulation} & \multicolumn{5}{c|}{Equivalence Checking}                                                                                                                                                                                                                               \\
        \cline{2-11}
                                 & \multicolumn{2}{c|}{Precision}  & \multicolumn{3}{c|}{Time (s)}             & \multicolumn{2}{c|}{Precision} & \multicolumn{3}{c|}{Time (s)}                                                                                                                                                                       \\
        \cline{2-11}
                                 & \multicolumn{1}{c|}{Traj}       & \multicolumn{1}{c|}{Ours}                 & \multicolumn{1}{c|}{TN}        & \multicolumn{1}{c|}{Traj}     & \multicolumn{1}{c|}{Ours} & \multicolumn{1}{c|}{Traj} & \multicolumn{1}{c|}{Ours} & \multicolumn{1}{c|}{TN} & \multicolumn{1}{c|}{Traj} & \multicolumn{1}{c|}{Ours} \\
        \hline
        qft\_10                  & 1.14E-2                         & 6.83E-3                                   & 362.88                         & 1205.43                       & 392.43                    & 3.65E-3                   & 6.59E-3                   & 441.44                  & 2679.16                   & 473.92                    \\
        \hline
        vqe\_10                  & 1.21E-2                         & 6.55E-3                                   & 44.91                          & 190.31                        & 34.94                     & 8.68E-3                   & 1.41E-2                   & 26.18                   & 225.73                    & 38.91                     \\
        \hline
        vqe\_12                  & 1.25E-2                         & 9.25E-3                                   & 70.96                          & 426.44                        & 67.83                     & 1.17E-2                   & 1.59E-2                   & 63.55                   & 470.85                    & 82.31                     \\
        \hline
        inst\_4x4\_10            & 1.54E-2                         & 1.18E-2                                   & 517.79                         & 591.36                        & 103.79                    & 4.80E-3                   & 8.57E-3                   & 2247.35                 & 2384.72                   & 449.10                    \\
        \hline
        inst\_4x4\_20            & 1.05E-2                         & 8.54E-3                                   & 13940.58                       & 12832.67                      & 2837.43                   & 8.67E-3                   & 1.00E-2                   & 18761.50                & 18441.19                  & 3829.20                   \\
        \hline
        inst\_4x5\_10            & 1.26E-2                         & 1.89E-2                                   & 8877.66                        & 12148.67                      & 1802.03                   & 1.84E-2                   & 1.22E-2                   & 12985.97                & 8758.94                   & 2628.17                   \\
        \hline
        inst\_4x5\_20            & 5.07E-3                         & 6.54E-3                                   & TO                             & TO                            & 20398.64                  & 1.09E-2                   & 1.19E-2                   & TO                      & TO                        & 24193.53                  \\
        \hline
    \end{tabular}
\end{table*}

\begin{figure}[!t]
	\centering
	\subfloat[Simulation error\label{fig:noise_rate_sim}]{
		\includegraphics[width=0.23\textwidth]{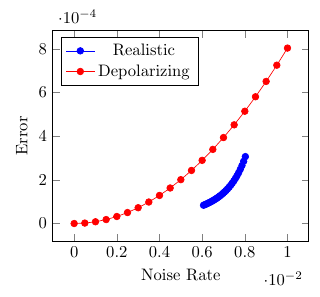}
	}
	\hfill
	\subfloat[Equivalence checking error\label{fig:noise_rate_eqc}]{
		\includegraphics[width=0.23\textwidth]{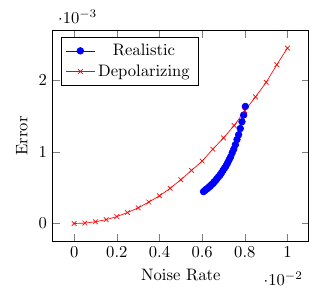}
	}
	\caption{Error at different noise rate}
	\label{fig:noise_rate}
\end{figure}

\subsection{Evaluation on NISQ Noise Model}
\revision{Finally, we evaluate our method on the NISQ noise model as introduced in Section~\ref{subsubsec:noise_model}. We analyzed the runtime and accuracy of our method relative to the quantum trajectories method, using the TN-based method as the baseline for accuracy evaluation. As in the comparison of single-qubit noise, we set the sample number of the quantum trajectories method to achieve a similar error bound as our method. It is worth noting that the NISQ noise model poses a significant challenge to simulation and equivalence checking, as it introduces substantial entanglement into the circuits. The large number of noise operators renders the required sampling numbers impractical for both methods. The results, shown in Table~\ref{table:crosstalk}, indicate that our approximation algorithm effectively handles the crosstalk noise model while achieving a strong balance between accuracy and efficiency. Furthermore, our algorithm proves more efficient than the quantum trajectories method, which demands a large number of samples to reach comparable accuracy.}

\section{Conclusion}

This paper presents a novel algorithm for efficiently simulating and checking the equivalence of noisy quantum circuits. Our approach represents noisy quantum circuits as tensor networks and leverages singular value decomposition (SVD) for approximation, achieving significantly faster performance than existing methods while maintaining the same level of accuracy. The algorithm, implemented using Google's TensorNetwork library, is demonstrated on various benchmark circuits with realistic noise models, including per-gate errors and multi-qubit crosstalk. Its effectiveness is highlighted by the successful simulation and equivalence checking of Quantum Approximate Optimization Algorithm (QAOA) circuits with approximately 200 qubits and 20 noise channels.

Further enhancements could be achieved by optimizing for specific circuit topologies (e.g., QAOA, VQE) to further reduce computational costs. Practically, this work has significant potential to enhance quantum Software Development Kits (SDKs) like Qiskit and Cirq for circuit optimization, hardware benchmarking, and prototyping large-scale quantum machine learning models. Given its scalability, we also anticipate our algorithm's integration into automatic test pattern generation (ATPG) programs~(e.g.,~\cite{fang2008fault,paler2012detection,bera2017detection, chen2024automatic}) for verifying and detecting manufacturing defects in large-scale quantum processors.

\section*{Acknowledgments}
This work was partially supported by the Innovation Program for Quantum Science and Technology (Grant
No. 2024ZD0300500), the Youth Innovation Promotion Association, Chinese Academy of Sciences (Grant No. 2023116), the National
Natural Science Foundation of China (Grant No. 62402485), the Young Elite Scientists Sponsorship Program, China Association for Science and Technology.

\printbibliography
\begin{IEEEbiography}[{\includegraphics[width=1in,height=1.25in,clip,keepaspectratio]{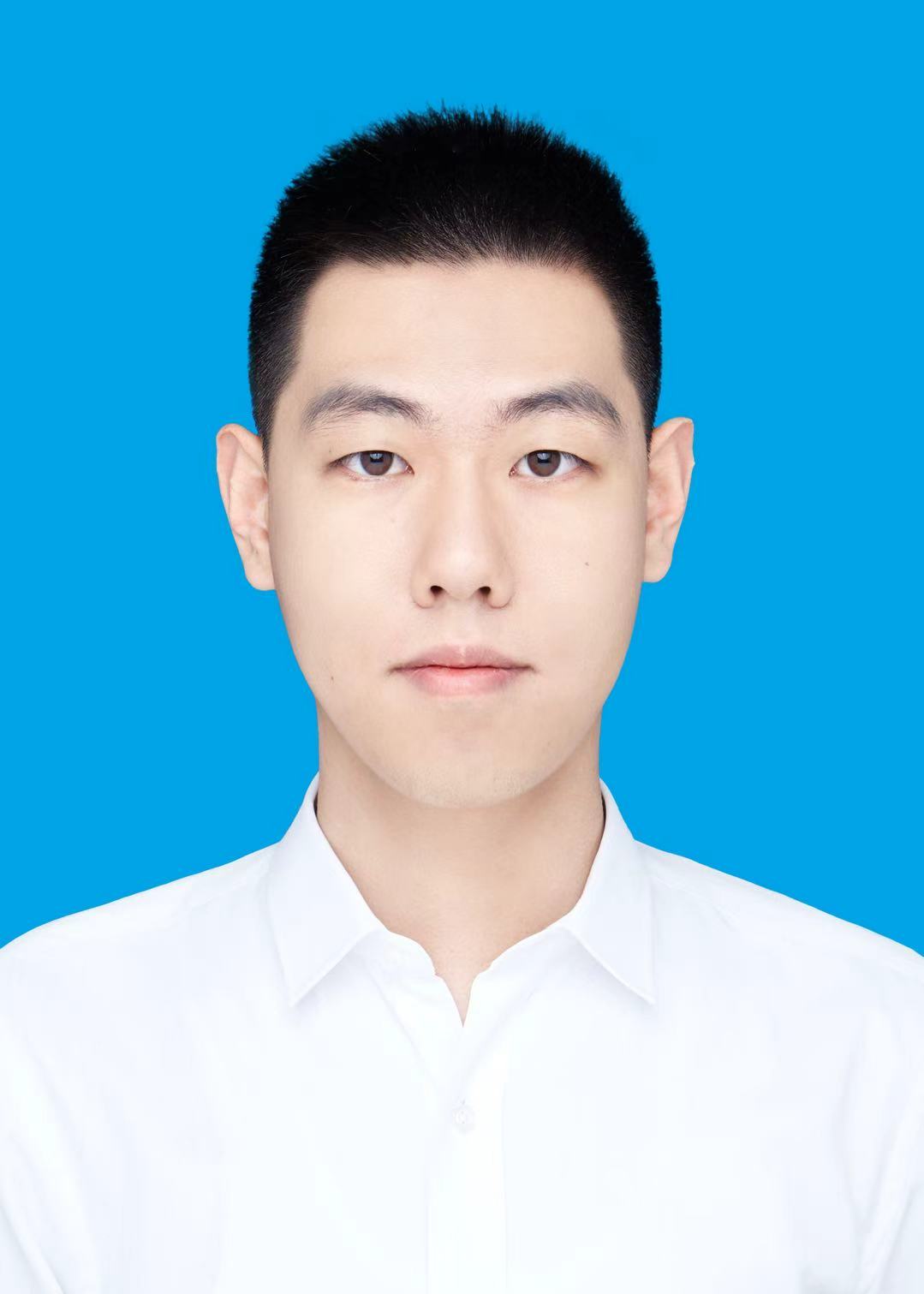}}]{Mingyu Huang}
received his BSc degree of Computer Science and Technology and Mathematics at Beijing Normal University in 2020. He is currently a Ph.D student at the Institute of Software, Chinese Academy of Sciences, advised by Prof. Mingsheng Ying. His research is about classical simulation of quantum circuits and quantum circuit optimization.
\end{IEEEbiography}

\vspace{-1cm}

\begin{IEEEbiography}[{\includegraphics[width=1in,height=1.25in,clip,keepaspectratio]{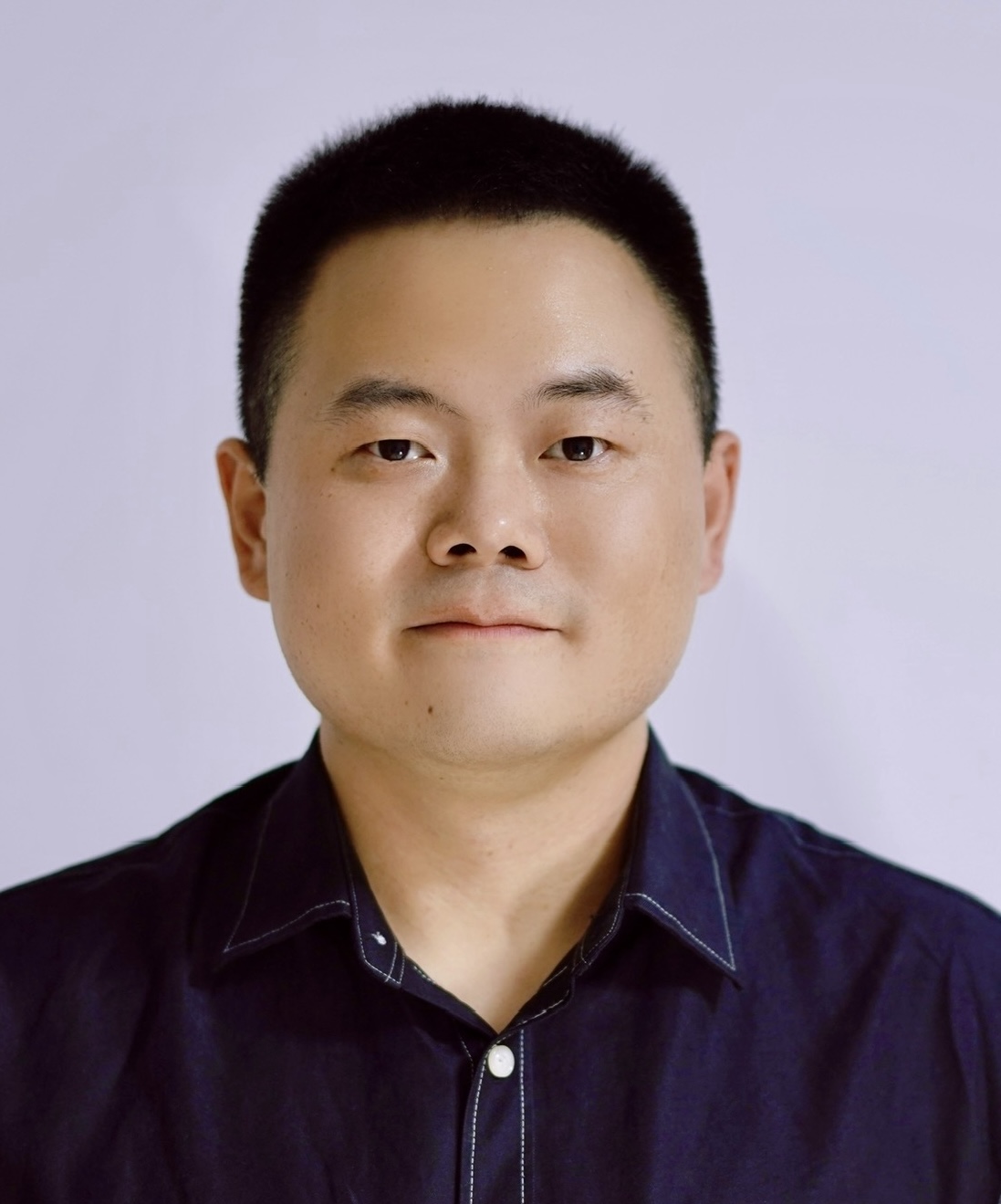}}]{Ji Guan}
	received the B.Sc. degree in mathematics from Sichuan University, China, in 2014, and the Ph.D. degree in computer science from the
	University of Technology Sydney, Australia, in 2018. He is currently a Faculty Researcher at the Institute of Software, Chinese Academy of Sciences. His primary research interests are trustworthy quantum machine learning and model checking quantum systems.
\end{IEEEbiography}

\vspace{-1cm}

\begin{IEEEbiography}
[{\includegraphics[width=1in,height=1.25in,clip,keepaspectratio]{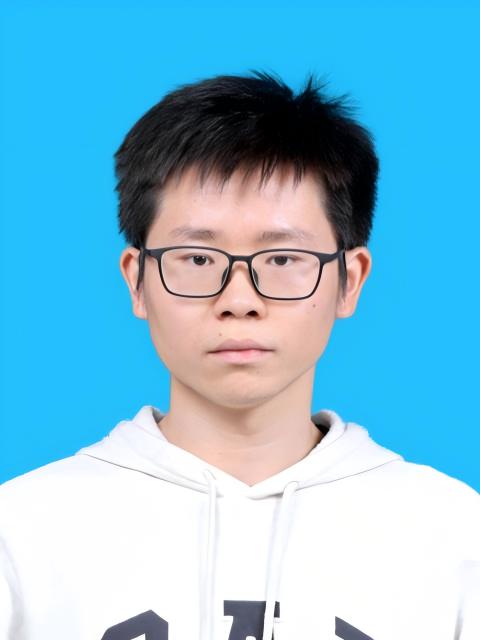}}]{Wang Fang}
received the B.Sc. degree in mathematics from Nanjing University, China, in 2018, and the Ph.D. degree in computer science from the Institute of Software, Chinese Academy of Sciences, China, in 2024. He is currently a Postdoctoral Research Associate at the School of Informatics, University of Edinburgh. His research interests include quantum programming languages, classical simulation of quantum circuits, and differentiable programming.
\end{IEEEbiography}

\vspace{-1cm}

\begin{IEEEbiography}
[{\includegraphics[width=1in,height=1.25in,clip,keepaspectratio]{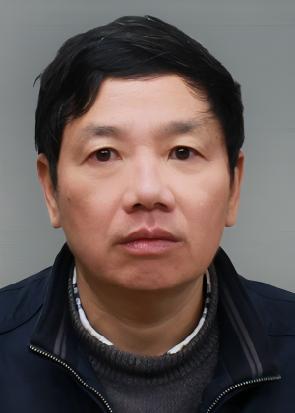}}]{Mingsheng Ying}
	received the Graduate degree from the Fuzhou Teachers College, Jiangxi, China, in 1981. He had been a Research Professor and the Deputy Director for Research of the Institute of Software, Chinese Academy of Sciences; and a Cheung Kong Professor with the Department of Computer Science and Technology, Tsinghua University, China. He is currently a Distinguished Professor at the Centre for Quantum Software and Information, University of Technology Sydney, Australia. He is the author of the books \textit{Model Checking Quantum Systems: Principles and Algorithms} (Cambridge University Press, 2021), \textit{Foundations of Quantum Programming} (Morgan Kaufmann, 2016), and \textit{Topology in Process Calculus: Approximate Correctness and Infinite Evolution of Concurrent Programs} (Springer-Verlag, 2001). His research interests are quantum computing, programming theory, and logics in artificial intelligence. Currently, he serves as the Co-Editor-in-Chief of ACM Transactions on Quantum Computing.
\end{IEEEbiography}

\end{document}